\documentclass[runningheads,envcountsect,envcountsame]{svjour3}
\PassOptionsToPackage{usenames,dvipsnames}{xcolor}
\usepackage[T1]{fontenc}
\usepackage{graphicx}
\usepackage{amsmath}
\usepackage{amssymb}
\usepackage[hidelinks]{hyperref}
\usepackage{cleveref}
\usepackage[color=blue!15]{todonotes}
\usepackage{mathtools}
\usepackage{booktabs}
\usepackage{enumitem}
\usepackage[noend]{algpseudocode}
\usepackage{algorithm}
\usepackage{algorithmicx}
\newcommand*{\Let}[2]{\State #1 $\gets$ #2}
\usepackage{tikz}
\usetikzlibrary{automata,positioning,shapes,shapes.geometric,decorations.pathmorphing}
\tikzset{counterexample/.style={color=black!15,line width=2mm,line cap=round,every loop/.style={}}}
\tikzset{fstate/.style={draw, circle}}
\usepackage{tabularx}
\usepackage{stmaryrd}
\usepackage{pgfplots}
\usepackage{subcaption}

\usepackage[labelfont=bf]{caption}

\usepackage[appendix=inline]{apxproof}

\input{fix-counter}

\newcommand{\ts}{\mathcal{T}}
\newcommand{\tsStates}{S}
\newcommand{\tsTrans}{\mathord{\rightarrow}}
\newcommand{\tsTransRel}{\rightarrow}
\newcommand{\tsInit}{s_{\mathit{init}}}
\newcommand{\tsTuple}{(\tsStates, \tsTrans, \tsInit)}

\newcommand{\tsRun}{\rho}
\newcommand{\tsRuns}{\mathit{Runs}}

\newcommand{\tsProp}{\Omega}

\newcommand{\eventually}{\lozenge}
\newcommand{\generally}{\square}

\newcommand{\parity}[1]{\mathit{Parity}\left(#1\right)}

\newcommand{\ggArena}{\mathcal{A}}
\newcommand{\ggPOne}{\text{Sat}}
\newcommand{\ggPTwo}{\text{Unsat}}
\newcommand{\ggSOne}{\tsStates_{\ggPOne}}
\newcommand{\ggSTwo}{\tsStates_{\ggPTwo}}
\newcommand{\ggStates}{\tsStates}
\newcommand{\ggTrans}{\mathord{\rightarrow}}
\newcommand{\ggTransRel}{\rightarrow}
\newcommand{\ggInit}{s_{\mathit{init}}}
\newcommand{\ggArenaTuple}{(\ggSOne, \ggSTwo, \ggTrans, \ggInit)}

\newcommand{\ggPlay}{\rho}
\newcommand{\ggPlays}{\mathit{Plays}}

\newcommand{\ggProp}{\Omega}
\newcommand{\ggGame}{\mathcal{G}}
\newcommand{\ggTuple}{(\ggArena, \ggProp)}

\newcommand{\ggStrat}{\sigma}
\newcommand{\ggStratOne}{\sigma_{\ggPOne}}
\newcommand{\ggStratTwo}{\sigma_{\ggPTwo}}
\newcommand{\ggIndPlay}{\ggPlay} 
\newcommand{\ggValue}{\mathit{val}}
\newcommand{\ggWin}{\mathit{Win}}

\newcommand{\tsGameArena}[2]{{#1}[#2]}

\newcommand{\shap}[1]{\mathit{Shap}_{#1}}
\newcommand{\coopGame}{\gamma}
\newcommand{\coopPlayers}{\mathcal A}
\newcommand{\coopPlayer}{A}

\newcommand{\engrave}{\mathit{Engrave}}

\newcommand{\pesGG}{\mathit{\mathcal{G}_{\mathit{pes}}}}
\newcommand{\optGG}{\mathit{\mathcal{G}_{\mathit{opt}}}}

\newcommand{\resp}[3]{\mathit{Resp}[#1,#2,#3]}
\newcommand{\respOpt}[3]{\mathit{Resp}_\mathit{opt}[#1,#2,#3]}
\newcommand{\respPes}[3]{\mathit{Resp}_\mathit{pes}[#1,#2,#3]}
\newcommand{\respName}{\mathit{Resp}}
\newcommand{\respOptName}{\mathit{Resp}_\mathit{opt}}
\newcommand{\respPesName}{\mathit{Resp}_\mathit{pes}}

\newcommand{\respGroupedName}{\mathit{Resp}^{\stateGroups}}

\newcommand{\stateGt}{\succ}
\newcommand{\stateLt}{\prec}
\newcommand{\stateGeq}{\succcurlyeq}
\newcommand{\stateLeq}{\preccurlyeq}

\newcommand{\winPlay}{\ggPlay_\mathit{win}}
\newcommand{\sbot}{s_\mathit{bottom}}
\newcommand{\sTop}{s_\mathit{top}}
\newcommand{\sskip}{s_\mathit{skip}}

\newcommand{\lowestReach}{\mathord{\downarrow}}
\newcommand{\lowestReachF}[1][F]{\mathord{\downarrow}_{#1}}

\newcommand{\Cmin}{C_\mathit{min}}

\newcommand{\stateGroups}{\mathfrak{S}}
\newcommand{\groupOf}[2][\stateGroups]{[#2]_{#1}}

\newcommand{\hasSP}{\mathit{HasBSP}}
\newcommand{\flatten}{\mathit{flatten}}

\newcommand{\p}{\textsc{P}}
\newcommand{\np}{\textsc{NP}}
\newcommand{\conp}{\textsc{CoNP}}
\newcommand{\npconp}{\np \! \cap \! \conp}
\newcommand{\pspace}{\textsc{PSpace}}
\newcommand{\exptime}{\textsc{ExpTime}}
\newcommand{\twoexptime}{\textsc{2ExpTime}}
\newcommand{\sharpp}{\texttt{\#}\textsc{P}}

\begin{document}
\title{Backward Responsibility in Transition Systems Beyond Safety}
\newcommand{\inst}[1]{\textsuperscript{#1}}
\newcommand{\orcidID}[1]{\textsuperscript{[#1]}}
\author{Christel Baier\inst{1,2}
	\and
Rio Klatt\inst{1,3}\and
Sascha Klüppelholz\inst{1}
\and
Max Korn\inst{1}
\and
Johannes Lehmann\inst{1,2}
\thanks{Authors are listed in alphabetical order. The authors are supported
	by the DFG
	through the Cluster of Excellence EXC 2050/2 (CeTI, project ID 390696704, as part of Germany's Excellence Strategy),
    the DFG grant 389792660 as part of TRR~248 (see \url{https://perspicuous-computing.science})
    and by BMBF (Federal Ministry of Education and Research) in DAAD project 57616814 (SECAI, School of Embedded and Composite AI) as part of the program Konrad Zuse Schools of Excellence in Artificial Intelligence.
}}

\authorrunning{C. Baier, R. Klatt, S. Klüppelholz, M. Korn and J. Lehmann}

\institute{\inst{1} Technische Universität Dresden, Dresden, Germany\and
\inst{2}Centre for Tactile Internet with Human-in-the-Loop (CeTI), Dresden, Germany\and 
\inst{3}University of Copenhagen, Denmark
\email{\{christel.baier,sascha.klueppelholz,max.korn,johannes\_alexander.lehmann\}@tu-dresden.de, gtn809@alumni.ku.dk}}

\maketitle
	
\setcounter{footnote}{0}

\vspace{-6em}
\noindent
\textbf{\textcolor{BrickRed}{This paper is an extended and revised version of \cite{Lehmann2025BeyondSafety}. The final paragraph of Section~\ref{sec:introduction} discusses the differences. The full proofs for \cite{Lehmann2025BeyondSafety} remain available in the first arXiv version, available at \mbox{\url{https://arxiv.org/abs/2506.05192v1}}.}}

\vspace{1em}

\begin{abstract}
	As the complexity of software systems rises, methods for explaining their behaviour are becoming ever-more important. When a system fails, it is critical to determine which of its components are responsible for this failure. Within the verification community, one approach uses graph games and Shapley values to ascribe a responsibility value to every state of a transition system. As this is done with respect to a specific failure, it is called \emph{backward responsibility}.
	
	This paper provides tight complexity bounds for the computation of backward responsibility values for reachability, Büchi and parity objectives. For Büchi objectives, a polynomial algorithm is given to determine the set of responsible states, i.e. states with positive responsibility value. To analyse systems that are too large for standard methods, the paper presents a novel refinement algorithm that iteratively finds the set of responsible states. Several heuristics are proposed to guide the refinement algorithm. Its utility is demonstrated with a tool that implements refinement in addition to several other responsibility computation techniques.

\end{abstract}

\section{Introduction}
\label{sec:introduction}
\label{sec:intro}

The increasing complexity of software systems makes them more opaque to their users. To alleviate this, there is growing demand for explaining which internal mechanisms lead to the exhibited behaviour of such systems.

This has given rise to the field of \emph{explainability}, which has produced a wide variety of techniques that apply to all stages of a system's use, from design time over execution time to inspection-time. These techniques can be forward-looking to identify potential causes of failures or backward-looking to explain a specific error scenario that was observed in a system. In recent years, there has been particular interest in explainable AI \cite{Xu2019ExplainableAI} due to the opacity of AI systems. Explainability techniques can help the developer identify faults and the user in understanding the behaviour of a system. Thus, explainability serves as a building block for transparency, trust and accountability in AI systems.

Recent trends towards explainable Formal Methods form one aspect of explainable AI. Techniques that verify whether a system fulfils or violates a specification are themselves often opaque. Providing additional information besides a yes-no-answer can both increase understanding of the underlying system and improve trust in the verification result. To prove that a model satisfies a specification, a verification tool may provide inductive invariants, deductive proofs \cite{MannaPnueliBook} or a witnessing subsystem \cite{Funke2020Witnessing}. To prove the violation of a specification, model checkers and theorem provers can produce counterexamples (see
e.g. \cite{Chockler2016Causality,FromVerificationToCausalityBaseExplications}
for an overview). In the taxonomy of \cite{ForwardBackward}, proofs that a specification is satisfied are forward-looking explications, as they take all possible behaviours of a model into account, while counterexamples and explications derived from them are backward-looking as they aim to explain a specific error scenario.

Counterexamples for regular linear-time properties are finite prefixes of computations, representing an ultimately periodic path that starts in an initial state and eventually enters a loop that is repeated ad infinity and that does not satisfy the property. Despite various techniques to generate short counterexamples \cite{Gastin2004MinimizationOfCounterexamples,Gastin2007MinimalCounterexample,Hansen2006MinimalCounterexamples}, counterexamples can still be very long, making it difficult to identify the reason for the misbehaviour. Motivated by this observation, multiple approaches have been proposed to explain counterexamples by identifying causal relationships or quantifying the degree of responsibility of system states and components. Some of these approaches rely on Halpern and Pearl's notions of actual causality and interventions \cite{HalpernPearlCauses,HalpernPearlExplanations}. The idea is to define causes as minimal sets of items that need to be modified (e.g. switching the truth value of atomic propositions in states) to avoid specification violations. The degree of responsibility of an individual item is then defined as $1/n$, where $n$ is the size of a smallest cause containing that item. This approach has been used to define the degree of responsibility of states in the forward-looking \cite{chockler2008causes} and the backward-looking \cite{Beer2009ExplainingCounterexamplesCausality} setting. Halpern and Pearl's notion of causality is also at the heart of \cite{Koelbl2020DynamicCauses,Caltas2018CausalityForGeneralLTL,LeitnerFischer2014SpinCause,LeitnerFischer2013CausalityChecking}.

Other approaches to identify causes in counterexamples rely on Stalnaker's and Lewis' semantics of counterfactuals in terms of most similar computations that avoid the error scenario. This is formalised by distance functions on computations \cite{Groce2006DistanceMetrics,ParreauxCounterfactualCausalityDistanceFunctions}. Delta debugging \cite{ZellerDeltaDebugging} uses a divide-and-conquer technique to generate the most similar passing run. The difference between passing and failing runs can also be used statistically to assess which components of a program are most suspicious \cite{LandsbergStatisticalFaultLocalisation}.

Another direction to explain counterexamples relies on Shapley values that have been introduced to measure the impact of individual players on the outcome of cooperative games \cite{ShapleyValueOriginal}. A recent discussions on the Shapley value can be found in \cite{ShapleyHandbook,ShapleyValueHandbook}. Such Shapley-value based approaches for formal responsibility notions have been first proposed for stochastic multi-player game structures \cite{GameTheoreticResponsibility} and for the forward-looking approach for temporal properties in transition systems \cite{ImportanceValuesTL}. In recent work, these approaches have been adapted to the backward-looking perspective by explaining counterexamples for safety properties \cite{BackwardRespTS}. This has been extended by lifting responsibility from individual states to higher-level concepts such as components, modules and actors \cite{BackwardRespTS}.

The current paper is in the line of these approaches, mostly of \cite{BackwardRespTS}.
Given a transition system as operational model, a temporal objective representing the specification and a counterexample path, a state $s$ has responsibility for the outcome if there is a set of states $C$ such that $C$ is not sufficient for satisfying the objective in the induced game, while $C \cup \{s\}$ is, i.e. $s$ is responsible whenever it made a difference in whether the objective is satisfied. Responsibility is then quantified using the Shapley value.


\paragraph{Contribution.}
Our contribution is threefold.
\begin{itemize}
	\item First, we provide complexity results for two decision problems (positivity and threshold) as well as the computation problem covering optimistic and pessimistic backward responsibility with reachability, Büchi and parity objectives (\Cref{sec:complexity}). This includes a polynomial-time algorithm that decides the optimistic positivity problem for Büchi objectives (\Cref{alg:opt_pos_buechi}).
	\item Secondly, we present a refinement algorithm (cf. \Cref{alg:refinement} in \Cref{sec:refinement}) for computing the set of states with positive responsibility. We propose a suitable heuristics and analyse its suitability with respect to different objective classes.
	\item Thirdly, we provide a new version of our existing responsibility tool~\cite{BackwardRespTS} that subsumes previous implementations and implements the refinement algorithm and support for Büchi objectives. \Cref{sec:experimental} experimentally evaluates the scalability and performance of the refinement algorithm and compares several heuristics. We show that for large models with few responsible states, the refinement algorithm can improve performance by several orders of magnitude. A case study shows how combining refinement and state grouping \cite{ActorBasedResponsibility} is used to locate an error in the model of a medical lab.
 \end{itemize}
 
\paragraph{Conference version.}
This work is based on the conference publication \emph{Backward Responsibility in Transition Systems Beyond Safety} \cite{Lehmann2025BeyondSafety} that appeared in the proceedings of the \emph{30th International Conference on Formal Methods for Industrial Critical Systems, FMICS 2025}. Compared to \cite{Lehmann2025BeyondSafety}, we now present full proofs, details omitted in the conference paper and provide additional examples and explanations. Beyond that, the new contribution focusses on the theory and practice of the refinement algorithm by
\begin{itemize}
	\item formally defining the \emph{frontier heuristics}, proving that it is well-defined and showing that it always identifies potentially responsible states for safety and reachability objectives,
	\item unifying the implementation of the refinement algorithm with previous implementations from \cite{BackwardRespTS,ActorBasedResponsibility,Lehmann2025BeyondSafety}, enabling their orthogonal contributions to be combined for the first time,
	\item expanding the analysis of refinement heuristics to new models and heuristics and
	\item demonstrating the combination of refinement with state grouping in a case study.
\end{itemize} 


\section{Preliminaries}

We briefly present the notations for transition systems, games on graphs and Shapley values. For more details, see e.g. \cite{ModelCheckingBook}.

\paragraph{Transition systems.}
A \emph{transition system} is a tuple $\ts = \tsTuple$, where $\tsStates$ is a finite set of \emph{states}, $\tsTrans \subseteq \tsStates \times \tsStates$ is the \emph{transition relation} (and we write $s \tsTransRel t$ for $(s, t) \in \tsTrans$) and $\tsInit \in \tsStates$ is the \emph{initial state}. A \emph{run} on $\ts$ is an infinite sequence of states $\tsRun = \tsRun_0 \tsRun_1 \ldots \in \tsStates^\omega$, where $\tsRun_0 = \tsInit$ and for all $i \in \mathbb N$, we have $\tsRun_i \tsTransRel \tsRun_{i+1}$. A run $\tsRun$ is \emph{lasso-shaped} if there are finite state sequences $u \in \tsStates^\ast$, $v \in \tsStates^+$ with $\tsRun = u v^\omega$. A lasso-shaped run is \emph{simple} if every state occurs at most once in $v$ and no state from $u$ also occurs in $v$. The set of all runs of $\ts$ is denoted by $\tsRuns(\ts)$. To simplify notations, we occasionally treat a run $\tsRun$ as a set of states, where $s \in \tsRun$ if there is an $i \in \mathbb N$ with $s = \tsRun_i$. An \emph{objective} is a set of runs $\tsProp \subseteq \tsRuns(\ts)$. A run $\tsRun$ \emph{fulfils objective $\tsProp$} if $\tsRun \in \tsProp$, otherwise, it \emph{violates $\tsProp$}. A transition system \emph{$\ts$ fulfils objective $\tsProp$} if $\tsRuns(\ts) \subseteq \tsProp$.

For $F \in \tsStates$, a \emph{safety objective} is given by $\generally \lnot F = \{ \tsRun \in \tsRuns(\ts) \mid \forall i \in \mathbb N \colon \tsRun_i \notin F \}$, a \emph{reachability objective} is given by $\eventually F = \{ \tsRun \in \tsRuns(\ts) \mid \exists i \in \mathbb N \colon \tsRun_i \in F\}$ and a \emph{Büchi objective} is given by $\generally\eventually F = \{ \tsRun \in \tsRuns(\ts) \mid \forall i \in \mathbb N \ \exists j > i \colon \tsRun_i \in F \}$. For a \emph{colouring function} $c \colon \tsStates \to \mathbb N$, a \emph{parity objective} is given by $\parity{c} = \{\tsRun \in \tsRuns(\ts) \mid \max \{ c(s) \mid s \in \mathit{Inf}(\tsRun) \} \text{ is even}\}$, where $\mathit{Inf}(\tsRun)$ denotes the states visited infinitely often in $\tsRun$.

\paragraph{Graph games.} A \emph{game arena} between two Players $\ggPOne$ and $\ggPTwo$ is a tuple $\ggArena = \ggArenaTuple$, where $\ggSOne$ and $\ggSTwo$ are the \emph{states controlled by Player~$\ggPOne$ and Player~$\ggPTwo$}, respectively (and we write $\ggStates \coloneqq \ggSOne \mathbin{\dot{\cup}} \ggSTwo$), $\ggTrans \subseteq \ggStates \times \ggStates$ is the \emph{transition relation} (and we write $s \ggTransRel t$ for $(s, t) \in \ggTrans$) and $\ggInit \in \ggStates$ is the \emph{initial state}. A \emph{play} on $\ggArena$ is an infinite sequence of states $\ggPlay = \ggPlay_0 \ggPlay_1 \ldots \in \ggStates^\omega$, where $\ggPlay_0 = \ggInit$ and for all $i \in \mathbb N$, we have $\ggPlay_i \ggTransRel \ggPlay_{i+1}$. The set of all plays of $\ggArena$ is denoted by $\ggPlays(\ggArena)$.

Given a transition system $\ts = \tsTuple$ and $C \subseteq \tsStates$, the \emph{corresponding two-player game arena} is defined by $\tsGameArena{\ts}{C} = (C, \tsStates \setminus C, \tsTrans, \tsInit)$.

An \emph{objective} is a set of plays $\ggProp \subseteq \ggPlays(\ggArena)$. A game is a tuple $\ggGame = \ggTuple$ of arena $\ggArena$ and objective $\ggProp$. A play $\ggPlay$ is \emph{winning} (with respect to $\ggProp$) if $\ggPlay \in \ggProp$. A \emph{strategy} for Player~$\ggPOne$ is a function $\ggStratOne \colon \ggSOne \to \ggStates$ such that $s \ggTransRel \ggStratOne(s)$ for all $s \in \ggSOne$. Strategies for Player~$\ggPTwo$ are defined analogously. A pair of strategies $\ggStrat = (\ggStratOne, \ggStratTwo)$ induces a play $\ggIndPlay = \ggPlay_0 \ggPlay_1 \ldots \in \ggStates^\omega$, where $\ggPlay_0 = \ggInit$, $\ggStratOne(\ggPlay_i) = \ggPlay_{i+1}$ for $\ggPlay_i \in \ggSOne$ and $\ggStratTwo(\ggPlay_j) = \ggPlay_{j+i}$ for $\ggPlay_j \in \ggSTwo$. A game is \emph{winning} for Player~$\ggPOne$ if there is a strategy $\ggStratOne$ of Player~$\ggPOne$ such that for all strategies $\ggStratTwo$ of Player~$\ggPTwo$, the induced play $\ggIndPlay$ satisfies the objective $\ggProp$ (and we say that $\ggIndPlay$ is \emph{winning for Player~$\ggPOne$}). Otherwise, the game is winning for Player~$\ggPTwo$. The \emph{value} of the game is defined as $\ggValue(\ggGame) = 1$ if $\ggGame$ is winning for Player~$\ggPOne$ and $\ggValue(\ggGame) = 0$ otherwise. The \emph{winning region} $\ggWin(\ggGame)$ is the set of states from which Player~$\ggPOne$ has a winning strategy.

Safety, reachability, Büchi and parity objectives for games are defined analogously to transition systems.

\paragraph{Shapley values.}
Let $\coopPlayers = \{ \coopPlayer_1, \ldots, \coopPlayer_n\}$ be a set of \emph{agents}. A \emph{simple cooperative game} is a monotonic function $\coopGame \colon 2^{\coopPlayers} \to \{0, 1\}$ and for $C \subseteq \coopPlayers$, we call $\coopGame(C)$ the \emph{payoff of coalition $C$}.
Shapley values \cite{ShapleyValueOriginal} are a well-known concept from cooperative game theory. They distribute the payoff of a cooperative game based on how much each agent contributed to the outcome. For a simple cooperative game $\coopGame$, the Shapley values are defined by the function $\shap{\coopGame} \colon \coopPlayers \to [0, 1]$ given by
\[ \shap{\coopGame}(\coopPlayer) = \sum_{C \subseteq \coopPlayers \setminus \{\coopPlayer\}} \frac{(|\coopPlayers| - |C| - 1)! \cdot |C|!}{|\coopPlayers|!} (\coopGame(C \cup \{\coopPlayer\}) - \coopGame(C)).\]

For $\coopPlayer \in \coopPlayers$ and $C \subseteq \coopPlayers \setminus \{\coopPlayer\}$, we call $(C, \coopPlayer)$ a switching pair if $\coopGame(C \cup \{\coopPlayer\}) - \coopGame(C) = 1$, i.e. if $\coopGame(C \cup \{\coopPlayer\}) = 1$ and $\coopGame(C) = 0$. Switching pairs correspond to the non-zero summands in the above sum, which implies that $\coopPlayer \in \coopPlayers$ has positive responsibility if and only there exists a switching pair $(C, \coopPlayer)$ for some $C \subseteq \coopPlayers$.

\section{Responsibility using the Shapley value}

This section revisits the definition of backward responsibility from~\cite{BackwardRespTS,ActorBasedResponsibility}, which extends forward responsibility as defined by \cite{ImportanceValuesTL}. Backward responsibility differs from forward responsibility by assessing responsibility with respect to a specific run $\tsRun$. The run $\tsRun$ is viewed as a counterexample illustrating a specification violation. Throughout this section, let $\ts = \tsTuple$ be a transition system with $\omega$-regular objective $\tsProp$ and simple lasso-shaped run $\tsRun = \tsRun_0 \tsRun_1 \ldots$ violating $\tsProp$.

The restriction to simple lasso-shaped runs ensures that $\tsRun$ can be reproduced by a positional strategy, i.e. it is not necessary to remember whether a state has been visited previously. The objective classes analysed in this paper can be refuted with simple lasso-shaped runs.

To incorporate the run $\tsRun$ into the system's analysis, \cite{BackwardRespTS} introduces an \emph{engraving} construction. For every state visited in $\tsRun$, all transitions are removed except for the one that follows $\tsRun$.
The construction is only applied to the states of $\tsRun$ that are not in a given set $C$.

\begin{definition}[Engraving construction]
	Let $\ts$ and $\tsRun$ be as above and let $C \subseteq \tsStates$. The \emph{engraved transition system} is given by $\engrave(\ts, \tsRun, C) = (\tsStates, \tsTrans', \tsInit)$ with $\tsTrans' = \{(s, t) \in \tsTrans \mid s \notin \tsRun \setminus C \} \hspace{0.5em} \cup\hspace{0.5em} \{ (\tsRun_i, \tsRun_{i+1}) \mid i \in \mathbb N \text{ and } \tsRun_i \notin C \}$.
\end{definition}

\begin{figure}[t]
    \centering
    \begin{subfigure}[t]{.49\textwidth}
        \centering
        \begin{tikzpicture}[shorten >=1pt,node distance=1.3cm,on grid,auto, state/.style={circle,inner sep=1pt}] 
            \node[state,initial,initial text=] (s0)   {$s_0$};
            \node[state] (s1) [above right=0.80cm and 0.70cm of s0]  {$s_1$};
            \node[state] (s2) [right= of s0]  {$s_2$};
            \node[state] (s3) [right= of s1]  {$s_3$};
            
            \node[state] (s4) [right= of s2]  {$s_4$};
            \node[state] (s5) [right= of s3]  {$s_5$};
            
            \node[state] (s6) [right= of s4]  {$\lightning$};
            
            \path (s0) edge[counterexample] node  {} (s2);
            \path (s2) edge[counterexample] node  {} (s4);
            \path (s4) edge[counterexample] node  {} (s6);
            \path (s6) edge [loop right, counterexample] node  {} (s63);
            
            \path[->] (s0) edge [bend right=15] node  {} (s1);
            \path[->] (s0) edge node  {} (s2);
            
            \path[->] (s1) edge [bend right=15] node  {} (s0);
            \path[->] (s1) edge node  {} (s3);
            
            \path[->] (s2) edge node  {} (s1);
            \path[->] (s2) edge node  {} (s3);
            \path[->] (s2) edge node  {} (s4);
            
            \path[->] (s3) edge node  {} (s5);
            
            \path[->] (s4) edge node  {} (s3);
            \path[->] (s4) edge node  {} (s5);
            \path[->] (s4) edge node  {} (s6);
            
            \path[->] (s5) edge node  {} (s6);
            
            \path[->] (s6) edge [loop right] node  {} (ss6);
        \end{tikzpicture}
        \caption{before engraving}
    \end{subfigure}%
    \begin{subfigure}[t]{.49\textwidth}
        \centering
        \begin{tikzpicture}[shorten >=1pt,node distance=1.3cm,on grid,auto, state/.style={circle,inner sep=1pt}] 
            \node[state,initial,initial text=] (s0)   {$s_0$};
            \node[state] (s1) [above right=0.80cm and 0.70cm of s0]  {$s_1$};
            \node[state] (s2) [right= of s0]  {$s_2$};
            \node[state] (s3) [right= of s1]  {$s_3$};
            
            \node[state] (s4) [right= of s2]  {$s_4$};
            \node[state] (s5) [right= of s3]  {$s_5$};
            
            \node[state] (s6) [right= of s4]  {$\lightning$};
            
            \path (s0) edge[counterexample] node  {} (s2);
            \path (s2) edge[counterexample] node  {} (s4);
            \path (s4) edge[counterexample] node  {} (s6);
            \path (s6) edge [loop right, counterexample] node  {} (s63);
            
            \path[->] (s0) edge node  {} (s2);
            
            \path[->] (s1) edge [bend right=15] node  {} (s0);
            \path[->] (s1) edge node  {} (s3);
            
            \path[->] (s2) edge node  {} (s4);
            
            \path[->] (s3) edge node  {} (s5);
            
            \path[->] (s4) edge node  {} (s3);
            \path[->] (s4) edge node  {} (s5);
            \path[->] (s4) edge node  {} (s6);
            
            \path[->] (s5) edge node  {} (s6);
            
            \path[->] (s6) edge [loop right] node  {} (ss6);
        \end{tikzpicture}
        \caption{after engraving}
    \end{subfigure}
    \caption{Engraving path $\tsRun = s_0 s_2 s_4 \lightning^\omega$ in a transition system with objective $\generally \lnot \lightning$, with coalition $C=\{s_4\}$. Outgoing transitions of the states of $\tsRun$ are removed if they do not follow $\tsRun$, except for $s_4$, because $s_4 \in C$.
    }
    \label{fig:engraving_example}
\end{figure}
\begin{example}
    Figure~\ref{fig:engraving_example} illustrates the engraving construction in a transition system with objective $\generally \lnot \lightning$ for path $\tsRun = s_0 s_2 s_4 \lightning^\omega$ and $C = \{s\}$. The outgoing transitions of $s_0$ and $s_2$ that diverge from $\tsRun$ are removed. State $s_4$ retains all outgoing transitions, as it is in $C$. This illustrates that the engraving construction enforces following the path, except in the states in $C$.
\end{example}

To determine the influence of some $C \subseteq \tsStates$ on fulfilling the objective, a graph game is constructed. For the behaviour of the remaining states $\tsStates \setminus C$, one can either ``optimistically'' assume they also try to fulfil the objectives or ``pessimistically'' assume that they try to prevent the objective from being fulfilled.

\begin{definition}[Graph game]
    Let transition system $\ts$, objective $\tsProp$ and run $\tsRun$ be as above and let $C \subseteq \tsStates$. Let $\ts' = \engrave(\ts, \tsRun, C)$ be the transition system obtained by engraving $\tsRun$ in $\ts$ except for the states in $C$. The \emph{pessimistic graph game} is defined as $\pesGG[\ts, \tsProp, \tsRun, C] = (\tsGameArena{\ts'}{C}, \tsProp)$ and the \emph{optimistic graph game} is defined as $\optGG[\ts, \tsProp, \tsRun, C] = (\tsGameArena{\ts'}{C'}, \tsProp)$, where $C' = C \cup \{ s \in \tsStates \mid s \notin \rho \}$.
    When $\ts$, $\tsProp$ and $\tsRun$ are clear from context, we say that \emph{$C$ is winning in the pessimistic setting} if Player~$\ggPOne$ wins $\pesGG[\ts, \tsProp, \tsRun, C]$ (and analogously for the optimistic case).
\end{definition}

In the optimistic graph game, only states with engraved behaviour are controlled by Player $\ggPTwo$. For simple lasso-shaped runs, Player $\ggPTwo$ thus has no decision power. 
Player $\ggPTwo$ may win nonetheless if the engraving prevents Player $\ggPOne$ from fulfilling the objective. The Shapley value now yields responsibility values.

\begin{example}
    Recall the example from Figure~\ref{fig:engraving_example}, with $C = \{s_4\}$. Even if Player $\ggPOne$ controls all states, reaching $\lightning$ is inevitable and thus, Player $\ggPTwo$ wins without controlling any states.
    On the other hand, for engraving with $C = \{s_2\}$, the outgoing transitions of $s_2$ are preserved, allowing Player $\ggPOne$ to force the loop $(s_0 s_2 s_1)^\omega$ and win.
\end{example}

\begin{definition}[Responsibility values]
	Let transition system $\ts$, objective $\tsProp$ and run $\tsRun$ be as above and let $C \subseteq \tsStates$. The \emph{responsibility values} $\respPes{\ts}{\tsProp}{\tsRun}$ and $\respOpt{\ts}{\tsProp}{\tsRun}$ are defined by the Shapley value of a function $2^\tsStates \to \{0, 1\}$ that maps $C \subseteq \tsStates$ to $\ggValue(\pesGG[\ts, \tsProp, \tsRun, C])$ in the pessimistic case and to $ \ggValue(\optGG[\ts, \tsProp, \tsRun, C])$ in the optimistic case.
\end{definition}

\begin{remark}
    Similar to the above, many other works also use the Shapley value for responsibility allocation \cite{ImportanceValuesTL,BackwardRespTS,FriedenbergBlameworthiness,YazdanpanahStrategicResponsibilityImperfectInformation,FromVerificationToCausalityBaseExplications} or rely on similar power indicies such as the \emph{Banzhaf power index} \cite{banzhaf1964weighted}. An alternative technique allocates responsibility based on the smallest switching pair \cite{chockler2004responsibility,aleksandrowicz2017computational}. This is done by determining the smallest $C\subseteq \tsStates$ such that $(C, s)$ is a switching pair and then defining the responsibility of $s$ as $\frac{1}{\left|C\right| + 1}$. However, the latter approach allocates more responsibility to a state with a single switching pair of size $3$ than to a state with hundreds of switching pairs of size $4$ -- the global view of the Shapley value, on the other hand, ensures that the latter state has higher responsibility.
\end{remark}

In the optimistic case, Player~$\ggPOne$ controls every state $s \in \tsStates \setminus \tsRun$ and therefore $\optGG[\ts, \tsProp, \tsRun, C] = \optGG[\ts, \tsProp, \tsRun, C \cup \{s\}]$. These states have no responsibility:

\begin{lemma}[Responsible states for optimistic responsibility]
	\label{lem:responsible_optimistic}
    Let $\ts$, $\tsProp$ and $\tsRun$ be as above. In the optimistic responsibility setting, only states on $\tsRun$ can have a positive responsibility value.
\end{lemma}
\begin{proof}
	This has been shown for safety objectives in \cite{BackwardRespTS}. The proof is unchanged for other objectives:
	Let $\ts = \tsTuple$ be a transition system with objective $\tsProp$ and run $\tsRun$ violating $\tsProp$. Let $s \notin \tsRun$ and $C \subseteq \tsStates$. Then $\engrave(\ts, \tsRun, C \cup \{s\}) = \engrave(\ts, \tsRun, C)$, as engraving only removes transitions from states in $C \cap \tsRun = \varnothing$. Therefore, 
	\begin{align*}
		& \optGG[\ts, \tsProp, \tsRun, C \cup \{s\}] \\
		= \quad & (\tsGameArena{\engrave(\ts, \tsRun, C \cup \{s\})}{C \cup \{s \in S \mid s \notin \tsRun\} \cup \{s\}}, \tsProp) \\
		= \quad & (\tsGameArena{\ts'}{C \cup \{s \in S \mid s \notin \tsRun\}}, \tsProp) \\
		= \quad & \optGG[\ts, \tsProp, \tsRun, C].
	\end{align*}
	Due to this, $(C, s)$ is not a switching pair and therefore, $s$ has no responsibility.
\end{proof}


\section{Complexity of responsibility computation problems}
\label{sec:complexity}

This section studies the computational complexity of determining the set of responsible states and computing their responsibility value. Throughout this section, let $\ts = \tsTuple$ be a transition system with $\omega$-regular objective $\tsProp$ and simple lasso-shaped run $\tsRun$ violating $\tsProp$. Let $\respName \in \{ \respOptName, \respPesName \}$.

{
    \renewcommand{\pspace}{\textsc{PSp}}
    \renewcommand{\exptime}{\textsc{ExpT}}
    \renewcommand{\twoexptime}{\textsc{2ExpT}}
    \newcommand{\gamesolving}[1]{} 
    \newcommand{\prop}[2][Prop.]{\,(#1\,\ref{#2})}
    \newcommand{\shiftColTwo}{\hspace{-0.70em}}
    \setlength{\tabcolsep}{1.7pt}
    \begin{table}[t]
        \caption{Overview of complexity results for the \textbf{pos}itivity and \textbf{comp}utation problem. The results are for completeness, except for those marked by $\in$, which are for inclusion in the respective class.}
        \label{tab:overview}
        \begin{tabular}{lcccccc}\toprule
            & \multicolumn{2}{c}{\hspace{-0.3em}Forward \cite{ImportanceValuesTL}} & \multicolumn{2}{c}{Backward, opt}& \multicolumn{2}{c}{Backward, pess}
            \\\cmidrule(lr){2-3}\cmidrule(lr){4-5}\cmidrule(lr){6-7}
            & \gamesolving{Game solving} Pos. & \hspace{-0.3em}Comp. &  Pos.  & Comp. &  Pos.  & Comp. \\\midrule
            Safety\hspace{1em}
            & \gamesolving{$\p$}
            $\np$
            & \shiftColTwo{}$\sharpp$
            & \hspace{-0.5em}$\in\!\p$ \cite{BackwardRespTS}
            & $\in\!\p$ \cite{BackwardRespTS}
            &  $\np$ \cite{BackwardRespTS}
            & $\sharpp$ \cite{BackwardRespTS} \\
            Reach.
            & \gamesolving{$\p$}
            $\np$
            & \shiftColTwo{}$\sharpp$
            & \hspace{-0.85em}$\in\!\p$\prop{prop:compl_reach_opt_pos}
            & $\in\!\p$\prop{prop:compl_reach_opt_comp}
            & $\np$\prop{prop:compl_reach_pes_pos}
            & $\sharpp$\prop{prop:compl_reach_pes_comp} \\
            Büchi
            & \gamesolving{$\p$}
            $\np$
            & \shiftColTwo{}$\sharpp$
            & \hspace{-0.93em}$\in\!\p$\prop{prop:compl_buechi_opt_pos}
            & $\in\!\sharpp$\prop{prop:compl_buechi_opt_comp}
            & $\np$\prop{prop:compl_buechi_pes_pos}
            &  $\sharpp$\prop{prop:compl_buechi_pes_comp} \\
            Parity
            & \gamesolving{\hspace{-0cm} $\in \!\! \npconp$}
            $\np$
            & \shiftColTwo{}$\sharpp$
            & \hspace{-1.05em}$\in \! \np$\prop{prop:compl_parity_opt_pos}
            & $\in\!\sharpp$\prop{prop:compl_parity_opt_comp}
            & $\np$\prop{prop:compl_parity_pes_pos}
            & $\sharpp$\prop{prop:compl_parity_pes_comp} \\
            \bottomrule
        \end{tabular}
    \end{table}
}
\begin{definition}[Responsibility problems]
    Given a transition system $\ts$, an $\omega$-regular objective $\tsProp$, a run $\tsRun$ violating $\tsProp$ and a state $s \in \tsStates$,
    \begin{itemize}
        \item
        the \emph{positivity problem} asks whether $\resp{\ts}{\tsProp}{\tsRun}(s) > 0$,
        \item
        the \emph{threshold problem} asks whether $\resp{\ts}{\tsProp}{\tsRun}(s) > t$ for a given threshold $t \in [0,1]$ and
        \item
        the task of the \emph{computation problem} is to compute $\resp{\ts}{\tsProp}{\tsRun}(s)$.
    \end{itemize}
\end{definition}

	Table~\ref{tab:overview} gives an overview of the complexity of these problems for different objective classes, both for optimistic and pessimistic background responsibility. For completeness, the figure also includes the results for forward responsibility as given in \cite{ImportanceValuesTL}. For forward responsibility, the input of the problems does not include a run $\tsRun$ violating $\tsProp$ and the engraving step is omitted. The remaining steps are the same for the forward and backward view.
	
	Recall that \sharpp{} is a class of counting problems that ask how many accepting runs a non-deterministic Turing machine with polynomial runtime have (see e.g. Chapter 9.1 in \cite{arora2009computational} for a formal introduction to \sharpp{}). To allow analysing the complexity of the computation problem, we treat it as a counting problem in the following, where the goal is to count the number of switching pairs (instead of computing the responsibility value).

\subsection{Determining the set of responsible states}

Deciding the optimistic positivity problem for safety objectives requires polynomial time \cite{BackwardRespTS}. This section extends this complexity result to reachability and Büchi objectives.

\begin{proposition}
\label{prop:compl_reach_opt_pos}
For reachability objectives, the optimistic positivity problem can be decided in polynomial time.
\end{proposition}

\begin{proof}
       {\newcommand{\winRun}{\tsRun_\mathit{win}}
       	Let $\tsProp = \eventually F$ be a reachability objective. In the following, we consider optimistic responsibility. We first show that, if $(C, s)$ is a switching pair for some $C \subseteq \tsStates$, then $(\varnothing, s)$ is also a switching pair.
        
        Let $(C, s)$ be a switching pair. Assume that $(\varnothing, s)$ is not a switching pair. Then there must be some non-empty $C' \subseteq C$ such that $(C', s)$ is a switching pair and $C'$ is minimal. By Lemma~\ref{lem:responsible_optimistic}, all states in $C'$ are on $\tsRun$. As $(C', s)$ is a switching pair, Player~$\ggPOne$ wins $\optGG[\ts, \eventually F, \tsRun, C' \cup \{s\}]$. In optimistic games, Player~$\ggPTwo$ does not control any states with more than one outgoing transition, so the winning strategy of Player~$\ggPOne$ induces a finite run prefix $\winRun$ that reaches $F$. Let $s' \in C' \cup \{s\}$ be the last state in $\winRun$ that also occurs in $C$. Because $s'$ is on $\tsRun$ and there is a path from $s'$ to $F$ that does not visit $\tsRun$ again, $\{s'\}$ is winning on its own. Therefore, $(\varnothing, s')$ is a switching pair and the assumption that $C'$ is minimal was wrong.
    }
    
    Therefore, to decide whether a state $s$ has positive responsibility, it is sufficient to verify that that Player~$\ggPOne$ wins $\optGG[\ts, \eventually F, \tsRun, \{s\}]$. This can be checked in polynomial time for reachability objectives.
\end{proof}

Now let $\tsProp = \generally \eventually F$ be a Büchi objective. Unlike for reachability objectives, $(\varnothing, s)$ is no longer necessarily a switching pair for a responsible state $s$, as shown in Figure~\ref{fig:buechi_example}: $(\{s_0\}, s_1)$ is a switching pair, but $(\varnothing, s_1)$ is not: Controlling $s_0$ and $s_1$ allows Player~$\ggPOne$ to produce the accepting run $(s_0 s_5 s_1 s_4)^\omega$, while controlling only $s_1$ makes $s_5$ unreachable due to the engraving. The reachability positivity algorithm is therefore not applicable to transition systems with a Büchi objective. Nonetheless, there is an efficient algorithm for the optimistic positivity problem.

\begin{figure}[t]
    \centering
    \begin{minipage}{.40\textwidth}
        \centering
        \begin{tikzpicture}[shorten >=1pt,node distance=1.1cm,on grid,auto, state/.style={circle,inner sep=1pt}] 
            \node[state,initial,initial text=] (s0)   {$s_0$};
            \node[state] (s1) [right=1.8cm of s0] {$s_1$};
            \node[state,fstate] (s2) [right=0.9cm of s1] {$s_2$};
            \node[state] (s3) [right=0.9cm of s2] {$s_3$};
            \node[state] (s4) [above right=0.75cm and 0.9cm of s0] {$s_4$};
            \node[state,fstate] (s5) [below right=0.75cm and 0.9cm of s0] {$s_5$};
            
            \path (s0) edge[counterexample] node  {} (s1);
            \path (s1) edge[counterexample] node  {} (s2);
            \path (s2) edge[counterexample] node  {} (s3);
            \path (s3) edge [loop below, counterexample] node  {} (s3);
            
            \path[->] (s0) edge node  {} (s1);
            \path[->] (s1) edge node  {} (s2);
            \path[->] (s2) edge node  {} (s3);
            \path[->] (s0) edge node  {} (s5);
            \path[->] (s4) edge node  {} (s0);
            \path[->] (s5) edge node  {} (s1);
            \path[->] (s1) edge[] node  {} (s4);
            \path[->] (s4) edge[bend left] node  {} (s1);
            \path[->] (s2) edge [loop below] node  {} (s2);
            \path[->] (s3) edge [loop below] node  {} (s3);
            
            \node[state,fstate,fill=white] (s2dummy) [right=0.9cm of s1] {$s_2$};
        \end{tikzpicture}
    \end{minipage}%
    \begin{minipage}{.59\textwidth}
        \centering
        \setlength{\tabcolsep}{5pt}
        \begin{tabular}{lrrrrrr}
            & $s_0$ & $s_1$ & $s_2$ & $s_3$ & $s_4$ & $s_5$\\
            \midrule
            Optimistic & $1/6$ & $1/6$ & $2/3$ & $0$ & $0$ & $0$ \\
            Pessimistic & $1/12$ & $1/12$ & $3/4$ & $0$ & $1/12$ & $0$ \\
            \bottomrule
        \end{tabular}
    \end{minipage}
    \caption{A transition system with Büchi objective $\generally\eventually \{s_2, s_5\}$ and violating run $\tsRun s_0 s_1 s_2 s_3^\omega$. The corresponding optimistic and pessimistic responsibility values are depicted on the right. Unlike for safety and reachability, not all states with positive optimistic responsibility have the same amount of responsibility.
    }
    \label{fig:buechi_example}
\end{figure}

{
    \begin{proposition}
        \label{prop:compl_buechi_opt_pos}
        The optimistic positivity problem for Büchi objectives is decidable in polynomial time.
    \end{proposition}
    
    \begin{proofsketch}
    	The proof relies on the fact that any non-singleton minimal winning coalition has a specific shape (illustrated in \Cref{fig:buechi_pos_struct} and formalised in Lemma~\ref{lem:buechi_opt_properties}). This induces an polynomial algorithm (cf. \Cref{alg:opt_pos_buechi}) that first checks whether a singleton minimal winning coalition exists and, if not, attempts to construct a suitable non-singleton minimal winning coalition. Correctness is shown in Proposition~\ref{prop:opt_pos_buechi_correctness}.
    \end{proofsketch}
    
    For the proof, the following notation helps relate the states of $\tsRun$. For states $s, t \in \tsRun$ on the run, $s \stateLeq t$ if $t$ is reachable from $s$ in $\engrave(\ts, \tsRun, \tsStates)$ and $s \stateLt t$ if $s \stateLeq t$ and not $t \stateLeq s$. A simple reachability analysis computes $\stateLt$ and $\stateLeq$ in polynomial time. We define the \emph{lowest reachable state from~$s$} as $\lowestReach(s)=\tsRun_i$ such that Player~$\ggPOne$ can reach $\tsRun_i$ in $\optGG[\ts, F, \tsRun, \{ s\}]$ and $i$ is minimal. For $F \subseteq \tsStates$, the \emph{lowest reachable state from $s$ through $F$} is defined as $\lowestReachF(s) = \tsRun_j$ such that there is an $f \in F$ for which Player~$\ggPOne$ can reach $f$ from $s$ and $t$ from $f$ in $\optGG[\ts, F, \tsRun, \{ s\}]$ and $j$ is minimal. If no such state $\tsRun_j$ exists, then $\lowestReachF(s) = \bot$.

    \begin{figure}[t]
    	\centering
    	\pgfdeclarelayer{ce}
    	\pgfsetlayers{ce,main}
    	\begin{tikzpicture}[shorten >=1pt,node distance=1.0cm,on grid,auto, state/.style={rectangle,inner sep=1pt}] 
    		\node[state,initial,initial text=] (downsn) {$\lowestReach(s_n)$};
    		\node[state] (sbot) [right=1.5cm of downsn] {$\sbot$};
    		\node[state] (downs2) [right=1.3cm of sbot] {};
    		\node[state] (sn) [right=0.5cm of downs2] {};
    		\node[state] (downs1) [right=3.0cm of sbot] {$\lowestReach(s_1)$};
    		\node[state] (s2) [right=of downs1] {$s_2$};
    		\node[state] (downstop) [right=of s2] {$\lowestReach(\sTop)$};
    		\node[state] (s1) [right=of downstop] {$s_1$};
    		\node[state] (upsbot) [right=1.5cm of s1] {$\lowestReachF(\sbot)$};
    		\node[state] (stop) [right=1.5cm of upsbot] {$\sTop$};
    		\node[state,fstate] (sf) [below right=0.6cm and 0.8cm of downs1] {$s_F$};
    		
    		\begin{pgfonlayer}{ce}
    			\path[] (downsn) edge[counterexample] node  {} (sbot);
    			\path[] (sbot) edge[counterexample] node  {} (downs1);
    			\path[] (downs1) edge[counterexample] node  {} (s2);
    			\path[] (s2) edge[counterexample] node  {} (downstop);
    			\path[] (downstop) edge[counterexample] node  {} (s1);
    			\path[] (s1) edge[counterexample] node  {} (upsbot);
    			\path[] (upsbot) edge[counterexample] node  {} (stop);
    			\path[] (stop) edge[loop right,counterexample] node  {} (stop);
    		\end{pgfonlayer}
    		
    		\path[->] (downsn) edge node  {} (sbot);
    		\path[->,dash pattern=on 14pt off 3pt on 3pt off 3pt on 3pt off 3pt on 3pt off 3pt on 3pt off 3pt on 3pt off 3pt on 14pt] (sbot) edge node  {} (downs1);
    		\path[->] (downs1) edge node  {} (s2);
    		\path[->] (s2) edge node  {} (downstop);
    		\path[->] (downstop) edge node  {} (s1);
    		\path[->] (s1) edge node  {} (upsbot);
    		\path[->] (upsbot) edge node  {} (stop);
    		\path[->] (stop) edge[loop right] node  {} (stop);
    		\path[->] (downsn) edge node  {} (sbot);
    		
    		\path[->] (stop) edge[bend right=27] node  {} (downstop);
    		\path[->] (s1) edge[bend right=38] node  {} (downs1);
    		\path[->, dash pattern=on 50pt off 3pt on 3pt off 3pt on 3pt off 3pt on 3pt off 3pt on 3pt off 3pt on 3pt off 3pt on 3pt off 3pt on 3pt off 3pt on 3pt off 3pt on 3pt off 3pt on 3pt off 3pt] (s2) edge[bend right=51] node  {} (downs2);
    		\path[->, , dash pattern=on 3pt off 3pt on 3pt off 3pt on 3pt off 3pt on 3pt off 3pt on 3pt off 3pt on 3pt off 3pt on 3pt off 3pt on 3pt off 3pt on 3pt off 3pt on 3pt off 3pt on 50pt] (sn) edge[bend right=36] node  {} (downsn);
    		
    		\path[->] (sbot) edge[bend right=10] node  {} (sf);
    		\path[->] (sf) edge[bend right=10] node  {} (upsbot);
    	\end{tikzpicture}
    	\caption{For Büchi objectives, minimal non-singleton winning coalitions matches the depicted layout. Each arrow $s \rightarrow t$ indicates that there is a path from $s$ to $t$, potentially visiting other states that are not depicted.
    	}
    	\label{fig:buechi_pos_struct}
    \end{figure}
    
    The algorithm for the optimistic positivity problem relies on the insight that if $C \cup \{s\}$ is winning and minimal,
    then either $C = \varnothing$ or the states are arranged as shown in Figure~\ref{fig:buechi_pos_struct}. We have $C \cup \{s\} = \{\sbot, \sTop, s_1, \ldots, s_n\}$, where
    \begin{itemize}
        \item $\sbot$ can reach some Büchi state $s_F$ and then continue to some state above $\lowestReachF(\sbot)$, from which $\sTop$ is reached by following $\tsRun$,
        \item $\sTop$ can jump down to an earlier state of $\tsRun$ to some state above $\lowestReach(\sTop)$, from which $s_1$ is reached by following $\tsRun$,
        \item $s_1$ can jump down to a state above $\lowestReach(s_1)$, from which $s_2$ is reached, and so on, until $\lowestReach(s_n)$ is reached, and
        \item from $\lowestReach(s_n)$, state $\sbot$ is reachable again by following $\tsRun$.
    \end{itemize}
    
    It is easy to see that these states are able to enforce the winning loop.
    Every $t \in \{\sTop, s_1, \ldots, s_n\}$ must ``skip'' some state that no other state in $C \cup \{s\}$ skips to ensure minimality of $C$.
    For example, in Figure~\ref{fig:buechi_pos_struct}, $\lowestReach(\sTop)$ is skipped only by the jump from $s_1$ to $\lowestReach(s_1)$, so without controlling $s_1$, it would be impossible to go from $\sTop$ to $\sbot$. State $\sbot$ is always necessary for winning to ensure that a Büchi state is reached. Minimality additionally required that no state in $C$ is winning on its own. The following lemma formalises this structure.
    
    \begin{lemma}[Properties of a minimal coalition]
        \label{lem:buechi_opt_properties}
        Let transition system $\ts$, objective $\generally \eventually F$ and run $\tsRun$ be as above. If $C \subseteq \tsStates$ is a minimal winning coalition and $|C| > 1$, then
        \begin{enumerate}[label=(\alph*)]
            \item \label{lem:buechi_opt_properties_loop}
            Player~$\ggPOne$ has a winning play $\winPlay = u v^\omega$ in $\optGG[\ts, \omega, \tsRun, C]$ such that $u, v \in \tsStates^\ast$ and every state $s \in C$ occurs exactly once in the loop $v$ in $\winPlay$,
            \item \label{lem:buechi_opt_properties_sbot} there exists exactly one state $\sbot \in C$ such that $\lowestReachF[F](\sbot) \stateLeq \sTop$ for some state $\sTop \in C$,
            \item \label{lem:buechi_opt_properties_sbot_lowest} we have $\sbot \stateLt s \stateLt \sTop$ for every state $s \in C \setminus \{\sbot, \sTop\}$ and
            \item \label{lem:buechi_opt_properties_skip} for every state $s_i \in C \setminus \{\sbot\}$, there exists a state $\sskip \in \tsRun$ with $\sbot \stateLt \sskip \stateLeq \lowestReachF[F](\sbot)$ such that $\lowestReach(s)  \stateLt \sskip{} \stateLeq s$ holds for $s$ and no other $s' \in C \setminus\{\sbot\}$.
        \end{enumerate}
    \end{lemma}
    
    
    \begin{proof}
        We first show \ref{lem:buechi_opt_properties_loop}. Because $C$ is a minimal winning coalition, Player~$\ggPOne$ wins $\optGG[\ts, F, \tsRun, C]$, so there must be a strategy $\ggStratOne$ for Player~$\ggPOne$. Because $\tsRun$ is simple lasso-shaped, no state controlled by Player~$\ggPTwo$ has more than one outgoing transition and therefore, $\ggStratOne$ induces a single run $\winPlay = u v^\omega$ for some $u, v \in \tsStates^\ast$.
        
        To show that every state in $C$ is part of the loop $v$, we first show that every state in $C$ is part of $\winPlay$, then that at least one state in $C$ is part of the loop $v$ and finally that every state must be.
        
        Every state in $C$ must occur in $\winPlay$: If a state $s \in C$ did not occur in $\winPlay$, then $C \setminus \{s\}$ would be sufficient to enforce $\winPlay$ and $C$ would not be minimal. (Note that this only works because all states of Player~$\ggPTwo$ have only one outgoing transition. Otherwise, additional states outside of $\winPlay$ might be required in $C$ to ensure that every possible action of Player~$\ggPTwo$ can be countered.)
        
        Now assume that no state in $C$ is part of $v$. In that case, no state of $v$ is in $\tsRun$ -- if such a state $s_\tsRun$ existed, then once $s_\tsRun$ were reached, execution would continue along $\tsRun$ and therefore, $\winPlay$ would not be winning for Player~$\ggPOne$. Therefore, $v$ does not intersect $\tsRun$. Let $s^\ast \in C$ be a state such that $t \stateLeq s^\ast$ for all $t \in C$. Then $\{ s^\ast \}$ is winning on its own, as execution can first follow along $\tsRun$ until $s^\ast$ is reached and then follow the same strategy that can enforce $\winPlay$. This contradicts the assumption that $|C| > 1$.
        
        \newcommand{\vpre}{v_\mathit{pre}}
        \newcommand{\vpost}{v_\mathit{post}}
        Finally, assume that there is some state $s^\ast \in C$ that does not occur in $v$.  Let $s_v \in C$ be a state that is part of $v$, let $v = \vpre s \vpost$ and $s_v = \tsRun_k$. Then $\tsRun_0 \cdots \tsRun_k \vpost (v^\omega)$ is a winning run that Player~$\ggPOne$ can enforce by controlling $C \setminus \{s^\ast\}$. This contradicts the assumption that $|C| > 1$.
        
        Therefore, every state in $C$ occurs in $v$.
        
        \newcommand{\vinf}{v_\mathit{inf}}
        We now show \ref{lem:buechi_opt_properties_sbot} by first showing that an $\sbot$ exists and then showing that only one $\sbot$ can exist. Due to \ref{lem:buechi_opt_properties_loop}, all states in $C$ are contained in the loop $v$ of $\winPlay$. As $\winPlay$ is winning, $v$ contains at least one Büchi state $s_F \in F$. Let $\vinf = t_0 \cdots t_k \in \tsStates^\ast$ be the unique infix of $v^\omega$ such that $\vinf$ contains the Büchi state $s_F$, $t_0, t_k \in C$ and $t_i \notin C$ for all $0 < i < k$.
        We set $\sbot = t_0$ and $\sTop = t_k$. They fulfil the condition $\lowestReachF[F](\sbot) \stateLeq \sTop$:
        
        By following $\vinf$, the Büchi state $s_F$ is reached from $\sbot$. Continuing along $\vinf$ reaches some state $\tsRun^\ast$ on $\tsRun$ and continuing further reaches $\sTop$. This does not require controlling states apart from $\sbot$, as $t_i \notin C$ for $0<i<t$. Therefore, Player~$\ggPOne$ can enforce reaching $\tsRun^\ast$ through $s_F$ from $\sbot$ in $\optGG[\ts, F, \tsRun, \{ \sbot\}]$. This implies that $\lowestReachF[F](\sbot) = \tsRun^\ast$ and $\tsRun^\ast \stateLeq \sTop$ because $\sTop$ is reached from $\tsRun^\ast$ by following $\vinf$.
        
        Now assume that there are two $\sbot \neq \sbot'$ that reach $\sTop$ and $\sTop'$, respectively, through $F$. Without loss of generality, let $\sTop \stateGeq \sTop'$. We first consider the case where $\sbot \stateGeq \sbot'$. Because $v$ forms a loop, Player~$\ggPOne$ can force to reach $\sbot'$ from $\sTop$ by following $v$ and this only requires control of $C \setminus \{\sbot'\}$. Due to $\sbot' \stateLeq \sbot$, Player~$\ggPOne$ can force to reach $\sbot$ from $\sbot'$ without controlling any states and it can reach $\sTop$ from $\sbot$ through $\lowestReach(\sbot)$. This forms a winning loop that Player~$\ggPOne$ can enforce by controlling $C \setminus \{\sbot'\}$, which leads to a contradiction. The case for $\sbot' \stateGeq \sbot$ is symmetrical.
        
        We now show \ref{lem:buechi_opt_properties_sbot_lowest} using a similar argument. 
        Assume by contradiction that there exists some state $s \in C \setminus \{ \sbot \}$ with $\sbot \not\stateLt s$. As $s$ and $\sbot$ are on $\tsRun$, it follows that $s \stateLeq \sbot$. Then Player~$\ggPOne$ can enforce a winning run by following $\winPlay$ until $s$ is reached and then always following $\tsRun$ from $s$ to $\sbot$ and following $v$ otherwise. For this, control of $C \setminus \{s\}$ is sufficient, contradicting the assumption that $C$ is minimal.
        
        Now assume by contradiction that a state $s \in C \setminus \{\sTop\}$ exists with $s \not\stateLt \sTop$. Then $\sTop \stateLeq s$. Then Player~$\ggPOne$ can enforce a winning run by following $\winPlay$ until $\sTop$ is reached and then always following $\tsRun$ from $\sTop$ to $s$ and following $v$ otherwise. For this, control of $C \setminus \{\sTop\}$ is sufficient, contradicting the assumption that $C$ is minimal.
        
        \newcommand{\sprev}{s_\mathit{prev}}
        \newcommand{\sprevlowest}{s^\downarrow_\mathit{prev}}
        \newcommand{\snext}{s_\mathit{next}}
        Finally, we show \ref{lem:buechi_opt_properties_skip}. Let $s \in C \setminus \{\sbot\}$, let $\sprev \in C$ be the first state of $C$ that precedes $s$ in $\winPlay$ and let $\snext \in C$ be the first state of $C$ that follows $s$ in $\winPlay$. If $\sprev = \sbot$, let $\sprevlowest = \lowestReachF(\sprev)$, otherwise, let $\sprevlowest = \lowestReach(\sprev)$. Then $\sprevlowest \stateLeq s$ and $\lowestReach(s) \stateLeq \snext$.
        
        We show that $\sprevlowest$ is a ``skip'' state for $s$:
        \begin{itemize}
            \item $\sprevlowest \stateLeq \lowestReachF(\sbot)$: For $\sprevlowest = \lowestReachF(\sbot)$, this holds trivially. Therefore, let $\sprevlowest \neq \lowestReachF(\sbot)$. Assume by contradiction that $\sprevlowest \stateLeq \lowestReachF(\sbot)$ does not hold and therefore $\lowestReachF(\sbot) \stateLeq \sprevlowest$. Then Player~$\ggPOne$ can form a winning loop by starting from $\sbot$, then going to $\lowestReachF(\sbot)$, continuing along $\tsRun$ until $s$ (because $\lowestReachF(\sbot) \stateLeq \sprevlowest \stateLeq s$), then following $\winPlay$ again until $\sbot$ is reached again. This loop is reachable by following $\tsRun$ and therefore forms a winning run that does not require control of $\sprev$. This contradicts the assumption that $C$ is minimal.
            \item $\sbot \stateLt \sprevlowest$: Assume by contradiction that this does not hold and therefore $\sprevlowest \stateLeq \sbot$. Then Player~$\ggPOne$ can form a winning loop by starting in $\sprev$, then going to $\sprevlowest$, continuing along $\tsRun$ until $\sbot$ and then continuing along $\winPlay$ until $\sprev$ is reached again. This loop is reachable by following $\tsRun$ and therefore forms a winning run that does not require control of $s$, contradicting the assumption that $C$ is minimal.
            \item $\sprevlowest \stateLeq s$: Because $\sprev$ is the first state from $C$ that precedes $s$ in $\winPlay$, it is possible to reach $s$ from $\sprev$ without controlling additional states. Therefore, the lowest state reachable from $\sprev$ is either $s$ or a state that is even lower.
            \item $\lowestReach(s) \stateLt \sprevlowest$: Assume by contradiction that this does not hold and therefore $\sprevlowest \stateLeq \lowestReach(s)$. Then Player~$\ggPOne$ can form a winning loop by following $\tsRun$ until$\sprev$, continuing down to $\sprevlowest$, then following $\tsRun$ to $\snext$ (because $\sprevlowest \stateLeq \lowestReach(s) \stateLeq \snext$) and then continuing along $\winPlay$ again until $\sprev$. This loop is reachable by following $\tsRun$ and therefore forms a winning run that does not require control of $s$, contradicting the assumption that $C$ is minimal.
            \item There is no $s' \in C \setminus\{\sbot\}$ with $\lowestReach(s') \stateLt \sprevlowest \stateLeq s'$: Assume by contradiction that such an $s'\in C$ exists. Note that $s' \neq \sprev$, because $\lowestReach(s') \stateLt \lowestReach(\sprev)$. We first consider the case where $\lowestReach(s) \stateLeq \lowestReach(s')$. In that case, a winning loop for Player~$\ggPOne$ is formed as follows: From $\sprev$, then continuing down to $\sprevlowest$, following $\tsRun$ to $s$ and continuing down to $\lowestReach(s)$. Let $\snext' \in C$ be the first state in $C$ that follows $s'$ in $\winPlay$. Then $\lowestReach(s) \stateLeq \lowestReach(s') \stateLeq \snext'$ and therefore, the run is continued by following $\winPlay$ from $\snext'$. Enforcing this loop does not require control of $s'$ and it is reachable by following $\tsRun$. This contradicts the assumption that $C$ is minimal.
            
            The case of $\lowestReach(s') \stateLeq \lowestReach(s)$ is analogous: In that case, $s$ is not necessary for forcing a win.
        \end{itemize}
    \end{proof}
    \begin{algorithm}[t]
        \caption{Deciding the optimistic positivity problem for Büchi objectives.
        }
        \label{alg:opt_pos_buechi}
        \begin{algorithmic}[1]
            \Function{IsResponsible}{$s$}
            \If{$s \notin \tsRun$}
            \State\Return{false}
            \EndIf
            \If{\Call{IsWinning}{$\{s\}$}}
            \State\Return{true}
            \EndIf
            \For{$\sTop \stateGeq \lowestReachF(s)$}
            \If{\Call{IsSbottom}{$s$, $\sTop$}}
            \State\Return{true}
            \EndIf
            \EndFor
            \For{$\sbot \stateLt s$ with $\lnot \Call{IsWinning}{\{\sbot\}}$}
            \For{$\sTop \stateGeq s$ with $\sTop \stateGeq \lowestReachF(\sbot)$}
            \For{$\sskip \in \tsRun$ with $\lowestReach(s) \stateLt \sskip \stateLeq s$ and $\sbot \stateLt \sskip \stateLeq \lowestReachF(\sbot)$}
            \If{\Call{Is$\mathrm{S}_i$}{$s$, $\sbot$, $\sTop$, $\sskip$}}
            \State\Return{true}
            \EndIf
            \EndFor
            \EndFor
            \EndFor
            \State\Return{false}
            \EndFunction
            \Statex
            \Function{IsSbottom}{$s$, $\sTop$}
            \Let{$C$}{$\{s, \sTop\}$}
            \For{$s' \in \tsRun$ with $s \stateLt s' \stateLeq \sTop$ and $\lnot \Call{IsWinning}{\{s\}}$}
            \If{$\lnot (\lowestReachF(s') \stateLeq \sTop)$}
            \Let{$C$}{$C \cup \{ s'\}$}
            \EndIf
            \EndFor
            \State\Return{\Call{IsWinning}{$C$}}
            \EndFunction
            \Statex
            \Function{IsSi}{$s$, $\sbot$, $\sTop$, $\sskip$}
            \Let{C}{$\{s, \sbot\}$}
            \For{$s' \in \tsRun$ with $\sbot \stateLt s' \stateLeq \sTop$ and $\lnot \Call{IsWinning}{\{s'\}}$}
            \If{$\lnot (\lowestReachF(s') \stateLeq \sTop)$ and $\lnot(\lowestReach(s') \stateLt \sskip \stateLeq s')$}
            \Let{$C$}{$C \cup \{s'\}$}
            \EndIf
            \EndFor
            \State\Return{\Call{IsWinning}{C}}
            \EndFunction
        \end{algorithmic}
    \end{algorithm}
    
    This lemma forms the base of Algorithm~\ref{alg:opt_pos_buechi} to decide the optimistic positivity problem für Büchi objectives in polynomial time. The algorithm first covers the cases that $s \notin \tsRun$ or that $s$ is winning on its own. The remainder of the algorithm checks whether $s$ is part of a non-singleton minimal winning coalition, handling $s = \sbot$ and $s \in \{\sTop, s_1, \ldots, s_n\}$ separately.
    
    In the former case (lines 6--8 and $\Call{IsSbottom}{}$), the algorithm iterates over every potential $\sTop$ and then checks whether there is a winning coalition $C$ with the given $\sbot$ and $\sTop$. For this, every state $s'$ between $\sbot$ and $\sTop$ is added to $C$ if it is neither winning on its own nor can reach $\sTop$ through $F$. If $\sTop$ were reachable through $F$, then $s'$ could form a winning loop by going through $F$ to $\sTop$ without ever visiting $\sbot$. Finally, it is checked whether $C$ is winning. If yes, then $(C \setminus \{s\}, s)$ is a switching pair as only $s$ can reach $\sTop$ through $F$. If not, then no such coalition can exist, as $C$ contains all states that are not winning already without $s$.
    
    The latter case (lines 9--13 and $\Call{IsSi}{}$) first involves iterating over all potential $\sbot$, $\sTop$ and skip states for $s$. The function $\Call{IsSi}{}$ then checks whether a coalition $C$ of states between $\sbot$ and $\sTop$ exists such that no $s' \in C \setminus \{ s\}$ is able to ``jump'' over $\sskip$ (line 24). Once again, states that are winning on their own or that can reach $\sTop$ through $F$ are excluded (line 23 and 24).
    If none of these checks succeed, then $s$ has no responsibility.
    Correctness of the algorithm is formally shown in the full version (\cite{ArxivVersion}, Proposition~
    B.1).
    
    \begin{propositionrep}
        \label{prop:opt_pos_buechi_correctness}
        Algorithm~\ref{alg:opt_pos_buechi} solves the optimistic positivity problem für Büchi objectives in polynomial time.
    \end{propositionrep}
    
    
    \begin{proof}
        We first show that Algorithm~\ref{alg:opt_pos_buechi} returns true if and only if $s$ has positive responsibility. If $s$ has positive responsibility, then there must exist a minimal winning coalition $\Cmin$ containing $s$. If $|\Cmin| = 1$, then $\Call{IsWinning}{s} = \mathrm{true}$ and therefore, the check in line~4 is satisfied, returns true. Otherwise, $|\Cmin| > 1$ and therefore, Lemma~\ref{lem:buechi_opt_properties} applies.
        
        If $s$ is $\sbot$ as in the lemma, then the loop in line~6 eventually finds the right $\sTop$, for which $\lowestReach(s) \stateLeq \sTop$ must hold due to Lemma~\ref{lem:buechi_opt_properties}~\ref{lem:buechi_opt_properties_sbot}. Due to Lemma~\ref{lem:buechi_opt_properties}~\ref{lem:buechi_opt_properties_sbot}, every state from $\Cmin$ fulfils the checks in line~18 and is therefore added to $C$. Therefore, $C$ contains the minimal winning coalition $\Cmin$ and is therefore itself is a winning coalition. The check in line~20 is satisfied and the algorithm returns \emph{true}.
        
        Similarly, if $s$ is some $s_i$, then the loops in lines 9 to 11 eventually find the right $\sbot$, $\sTop$ and $\sskip$. All states from $\Cmin$ fulfil the conditions in line~24 and therefore, $C$ is also a winning coalition. Thus, the algorithm returns \emph{true}.
        
        We now show that if Algorithm~\ref{alg:opt_pos_buechi} returns true, there exists a switching pair with $s$. If it returns \emph{true} in line~4, then a switching pair is trivially given by $(\varnothing, s)$. If \emph{true} is returned in line~8, a switching pair is given by $(C \setminus \{s\}, s)$: Due to the check in line~20, $C$ is known to be winning and thus, it is sufficient to show that $(C \setminus \{s\})$ is losing. For the coalition to win, due to Lemma~\ref{lem:buechi_opt_properties}~\ref{lem:buechi_opt_properties_sbot}, there must exist states $\sbot' \neq s$ and $\sTop'$ with $\lowestReachF(\sbot') \stateLeq \sTop'$. Due to line~17, $s' \stateLeq \sTop$ for every $s' \in C$ and thus $\sTop' \stateLeq \sTop$. Therefore, the check in line~18 fails for any $\sbot'$ with $\lowestReachF(\sbot') \stateLeq \sTop$, so such an $\sbot'$ cannot be contained in $C$. Therefore, $(C \setminus \{s\}, s)$ is a switching pair.
        
        Finally, if \emph{true} is returned in line~13, a switching pair is given by $(C \setminus \{s\}, s)$. Once again, it suffices to show that $C \setminus \{s\}$ is losing by showing that no minimal winning coalition $\Cmin \subseteq C$ exists. Assume by contradiction that such a $\Cmin$ existed. Then $|\Cmin| > 1$, as the checks in line~9 and line~23 ensure that no singleton winning coalition is included in $C \setminus \{s\}$. By Lemma~\ref{lem:buechi_opt_properties}~\ref{lem:buechi_opt_properties_sbot}, there are $\sbot$ and $\sTop$ in $C$. To form a loop containing a Büchi state, Player~$\ggPOne$ must first go from $\sbot$ to $\lowestReach(\sbot)$ via $F$. However, once $\lowestReach(\sbot)$ is reached, it is not possible to return to $\sbot$ by controlling $C \setminus\{s\}$: In line~11, $\sskip$ is chosen such that $\sbot \stateLt \sskip \stateLeq \lowestReachF(\sbot)$. Due to line~24, no state $s_\mathit{low}$ with $s_\mathit{low} \stateLt  \sskip$ can be reached from any state $s_\mathit{high} \in C \setminus \{s\}$ with $\sskip \stateLeq  s_\mathit{high}$. Therefore, $\sbot$ cannot be reached from $\lowestReach(\sbot)$ and $C \setminus \{s\}$ is therefore not sufficient for winning.
        
        We now show that Algorithm~\ref{alg:opt_pos_buechi} takes polynomial time: The algorithm contains nested loops that iterate over a series of states. Each iteration of the innermost loop only requires polynomial time and the loops iterate over polynomially many states. Therefore, the total runtime is polynomial in the number of states of the input transition system.
    \end{proof}
}

\subsubsection*{Optimistic responsibility for parity objectives.}

\begin{figure}[t]
    \centering
    \begin{tikzpicture}[shorten >=1pt,node distance=1.5cm,on grid,auto, state/.style={circle,inner sep=1pt}] 
        \node[state,initial,initial text=] (s0)   {$s_0\colon\!1$};
        \node[state] (s1) [right=of s0] {$s_1\colon\!1$};
        \node[state] (s2) [right=of s1] {$s_2\colon\!3$};
        \node[state] (s3) [right=of s2] {$s_3\colon\!1$};
        \node[state] (s4) [right=of s3] {$s_4\colon\!1$};
        \node[state] (s5) [below=0.9cm of s2] {$s_5\colon\!2$};
        
        \path (s0) edge[counterexample] node  {} (s1);
        \path (s1) edge[counterexample] node  {} (s2);
        \path (s2) edge[counterexample] node  {} (s3);
        \path (s3) edge[counterexample] node  {} (s4);
        \path (s4) edge[loop right,counterexample] node  {} (s4);
        
        \path[->] (s0) edge node  {} (s1);
        \path[->] (s1) edge node  {} (s2);
        \path[->] (s2) edge node  {} (s3);
        \path[->] (s3) edge node  {} (s4);
        \path[->] (s4) edge[loop right] node  {} (s4);
        \path[->] (s0) edge[bend right=15] node  {} (s5);
        \path[->] (s5) edge[bend right=15] node  {} (s4);
        
        \path[->] (s4) edge[bend right=28] node  {} (s1);
        \path[->] (s1) edge[bend right=30] node  {} (s3);
        \path[->] (s3) edge[bend right=28] node  {} (s0);
    \end{tikzpicture}
    \caption{The winning run in this transition system with parity objective requires jumping ``forward'' from $s_1$ to $s_3$, which is never required for Büchi objectives.
    }
\label{fig:parity_jump_forward}
\end{figure}

The algorithm for the optimistic positivity problem for Büchi objectives relied on the insight that there is always a lasso-shaped winning run $\tsRun = u v^\omega$ (for $u \in \tsStates^\ast$, $v \in \tsStates^+$), if there is a winning run at all. The loop $v$ of this run consists of a single ``jump forward'' (i.e. going from some state $s \in \tsRun$ via some Büchi state $f \in F$ to $s' \in \tsRun$ with $s \stateLt s'$ without following $\tsRun$). The remainder of the loop consists of ``backward jumps'' from $t \in \tsRun$ to $t' \in \tsRun$ with $t \stateGt t'$ and of segments of $\tsRun$, i.e. the loop $v$ contains only a single jump forward.

Figure~\ref{fig:parity_jump_forward} shows that a single jump forward is not sufficient for parity objectives. States $s_0$, $s_1$, $s_3$ and $s_4$ are winning by following the run $\tsRun = (s_0 s_5 s_4 s_1 s_3)^\omega$. This requires both a jump from $s_0$ via $s_5$ to $s_4$ and from $s_1$ to $s_3$. Following $\tsRun$ from $s_1$ to $s_3$ (via $s_2$) is not possible because $s_2$ with colour $3$ must not be visited. In the Büchi case, jumping ``forward'' is never necessary, as it is never necessary to avoid a state.

This observation shows that the algorithm for Büchi objectives does not work for parity objectives. The precise complexity of the optimistic positivity problem for partiy objectives is left open, but we provide an upper bound. The same upper bound also applies to pessimistic positivity problem for reachability, Büchi and parity objectives, where we additionally show completeness.

\begin{proposition}
    \label{prop:compl_parity_opt_pos}
    \label{prop:compl_reach_pes_pos}
    \label{prop:compl_buechi_pes_pos}
    \label{prop:compl_parity_pes_pos}
    The optimistic positivity problem for parity objectives is in NP. The pessimistic positivity problem for reachability, Büchi and parity objectives is NP-complete.
\end{proposition}

\begin{appendixproof}
	Inclusion in NP is shown by first guessing a coalition $C$ and then verifying that $(C, s)$ is a switching pair. For parity objectives, this is possible because solving parity games is known to be in $\text{NP} \cap \text{CoNP}$. NP-hardness is shown by a reduction from the forward case, which is known to be NP-hard \cite{ImportanceValuesTL}.
	
    We first show inclusion in NP. Let $\ts = \tsStates$ be a transition system with reachability, Büchi or parity objective $\tsProp$ and simple lasso-shaped run $\tsRun$ that violates $\tsProp$. Let $s \in \tsStates$. A non-deterministic Turing machine determines whether $s$ has positive responsibility in polynomial time as follows.
    
    First, the Turing machine non-deterministically guesses a set of states $C \subseteq \tsStates$. It next constructs the game $\pesGG[\ts, \tsProp, \tsRun, C]$ (or $\optGG[\ts, \tsProp, \tsRun, C]$ for the optimistic case for parity objectives) and non-deterministically guesses a strategy for Player~$\ggPTwo$. It then verifies that this strategy is winning for Player~$\ggPTwo$. Even though no polynomial algorithm is known for solving parity games, this step only requires polynomial time, as fixing the strategy of Player~$\ggPTwo$ turns the game into a transition system. Determining whether a transition system fulfils a parity objective takes polynomial time.
    
    Finally, the Turing machine constructs the parity game $\optGG[\ts, c, \tsRun, C \cup \{s\}]$, non-deterministically guesses a strategy for Player~$\ggPOne$ and verifies that this is a winning strategy for Player~$\ggPOne$. If this is the case, it accepts. The algorithm for pessimistic responsibility functions in the same way, except that it constructs $\pesGG$ instead of $\optGG$.
    
    This machine accepts if and only if $(C, s)$ forms a switching pair, as the algorithm directly checks that $C$ is no sufficient for winning and that $C \cup \{s\}$ is sufficient for winning. Therefore, it accepts if and only if $s$ has positive responsibility.
    
    
    \newcommand{\snotf}{s_{\lnot F}}    
    We now show NP hardness for the pessimistic case. For reachability objectives, we give a reduction from the forward case: In \cite{ImportanceValuesTL}, Proposition~IV.5, it is shown that the forward positivity problem is NP-hard for reachability objectives. Let transition system $\ts = \tsStates$, state $s \in \tsStates$ and reachability objective $F$ be the input of a forward positivity problem. We construct a new transition system $\ts' = (\tsStates \cup \{\tsInit', \snotf\}, \tsTrans', \tsInit')$ with $\tsTrans' = \tsTrans' \cup \{ (\tsInit', \tsInit), (\tsInit', \snotf), (\snotf, \snotf) \}$ and set $\tsRun = \tsInit' \snotf^\omega$. Then $(\ts', F, \tsRun, s)$ is an input for the pessimistic backward responsibility problem.
    
    We show that $s$ has positive pessimistic backward responsibility in $\ts', F, \tsRun$ if and only if it has positive forward responsibility in $\ts, F$. First, assume that $s$ has positive forward responsibility in $\ts, F$. Then there exists some $C \subseteq \tsStates$ such that $(C, s)$ is a switching pair. Therefore, Player~$\ggPOne$ has a winning strategy $\sigma$ when controlling $C \cup \{s\}$. Due to this, Player~$\ggPOne$ also has a winning strategy in $\pesGG[\ts', \tsRun, C \cup \{s\} \cup \{\tsInit'\}]$ by first going from $\tsInit'$ to $\tsInit$ and then following $\sigma$. This is possible because the engraving only removes transitions from $\snotf$. Therefore, $(C \cup \{\tsInit'\}, s)$ is a switching pair for $\ts', F, \tsRun$ and therefore, $s$ has positive pessimistic backward responsibility.
    
    Now assume that $s$ has positive pessimistic backward responsibility in $\ts, F, \tsRun$. Then there exists some $C \subseteq \tsStates$ such that $(C, s)$ is a switching pair. Therefore, Player~$\ggPOne$ has a winning strategy $\sigma$ in $\pesGG[\ts', \tsRun, C \cup \{s\}]$. Furthermore, $\tsInit' \in C$, as otherwise, the engraving would remove all transitions from $\tsInit'$ except for $\tsInit' \tsTransRel' \snotf$, therefore making $F$ unreachable. From $\tsInit'$, Player~$\ggPOne$ must first go to $\tsInit$, as this is the only way to reach $F$. Due to this, Player~$\ggPOne$ can follow the same strategy restricted to $\tsStates$ to win in the forward setting in $\ts, F$. Therefore, $(C \cap \tsStates, s)$ is a switching pair in the forward setting and therefore, $s$ has positive forward responsibility.
    
    NP-hardness for the pessimistic positivity problem for Büchi and parity objectives follows from the NP-hardness of the pessimistic positivity problem for reachability objectives.
\end{appendixproof}

\subsection{Computing the responsibility value}

Optimistic responsibility values can be computed efficiently for safety properties \cite{BackwardRespTS}. This relies on the insight that all responsible states are winning on their own for safety objectives. The following extends this result to reachability objectives.

\begin{lemma}[Responsibility values for reachability objectives]
    \label{lem:opt_resp_safe_reach}
    Let $\ts$ and $\tsRun$ be as above and let $\tsProp$ be a reachability objective. Then there is some rational $r \in [0, 1]$ such that $\respOpt{\ts}{\tsProp}{\tsRun}(s) \in \{0, r\}$ for every $s \in \tsStates$.
\end{lemma}

\newcommand{\respS}{\mathit{RespStates}}
\begin{proof}
    Let $\respS \subseteq \tsStates$ be the set of states that are winning on their own, i.e. $s \in \respS$ if and only if $(\varnothing, s)$ is a switching pair. For $s \in \respS$, $C \subseteq \tsStates$, the pair $(C, s)$ is switching if and only if $C \cap \respS = \varnothing$: $C \cup \{s\}$ is sufficient for winning as $s$ is already winning on its own and due to Proposition~\ref{prop:compl_reach_opt_pos}, $C$ is not sufficient for winning if and only if it contains no states that are winning on its down.
    
    Therefore, for any $s \in \respS$ and size $i$, there are ${|S| - |\respS| \choose i}$ different $C$ with size $i$ such that $(C, s)$ is a switching pair. As any responsible state has the same number of switching pairs of the same size, it also has the same responsibility as any other responsible state.
\end{proof}

This lemma induces a polynomial algorithm for computing responsibility values for reachability objectives, because responsible states can be identified in polynomial time by Proposition~\ref{prop:compl_reach_opt_pos}.

\begin{proposition}
    \label{prop:compl_reach_opt_comp}
    \label{prop:compl_reach_opt_thresh}
    Solving the optimistic computation and threshold problem for reachability objectives requires polynomial time.
\end{proposition}

\begin{proof}
    Because the optimistic positivity problem for reachability objectives requires polynomial time (Proposition~\ref{prop:compl_reach_opt_pos}), the set of responsible states $\respS$ can be computed in polynomial time. By Lemma~\ref{lem:opt_resp_safe_reach}, every state in $\respS$ has the same responsibility value. Due to the \emph{efficiency property} of the Shapley value \cite{ShapleyHandbook}, the sum of responsibilities must be $1$, unless the objective is unsatisfiable in the transition system, in which case $R$ is empty and all states have responsibility value $0$. Therefore, the responsibility value of every state in $\respS$ is $\frac{1}{\left|R\right|}$ and the responsibility value of every other state is $0$.
    
    To decide the threshold problem, this value is then compared to the given threshold.
\end{proof}

{        
    \newcommand{\sia}[1]{s_{#1}^{a}}
    \newcommand{\sib}[1]{s_{#1}^{b}}
    \begin{figure}[t]
        \centering
        \pgfdeclarelayer{ce}
        \pgfsetlayers{ce,main}
        \begin{tikzpicture}[shorten >=1pt,node distance=1.0cm,on grid,auto, state/.style={rectangle,inner sep=1pt}] 
            \node[state,initial,initial text=] (s0) {$s_0$};
            \node[state] (s1a) [right=of s0] {$\sia{1}$};
            \node[state] (s1b) [right=of s1a] {$\sib{1}$};
            \node[state] (s2a) [right=of s1b] {$\sia{2}$};
            \node[state] (s2b) [right=of s2a] {$\sib{2}$};
            \node[state] (sna) [right=3cm of s2b] {$\sia{n}$};
            \node[state] (snb) [right=of sna] {$\sib{n}$};
            \node[state,fstate] (sf) [below=0.65cm of s2b] {$s_f$};
            
            \node[state] (s3a) [right=of s2b] {};
            \node[state] (s3b) [right=of s3a] {};
            \node[state] (snminus1b) [left=of sna] {};
            \node[state] (snminus1a) [left=of snminus1b] {};
            
            \begin{pgfonlayer}{ce}
                \path (s0) edge[counterexample] node  {} (s1a);
                \path (s1a) edge[counterexample] node  {} (s1b);
                \path (s1b) edge[counterexample] node  {} (s2a);
                \path (s2a) edge[counterexample] node  {} (s2b);
                \path (s2b) edge[counterexample] node  {} (sna);
                \path (sna) edge[counterexample] node  {} (snb);
                \path (snb) edge[counterexample,loop right] node  {} (snb);
            \end{pgfonlayer}
            
            \path[->] (s0) edge node  {} (s1a);
            \path[->] (s1a) edge node  {} (s1b);
            \path[->] (s1b) edge node  {} (s2a);
            \path[->] (s2a) edge node  {} (s2b);
            \path[->,dash pattern=on 7pt off 3pt on 3pt off 3pt on 3pt off 3pt on 3pt off 3pt on 3pt off 3pt on 3pt off 3pt on 3pt off 3pt on 3pt off 3pt on 3pt off 3pt on 3pt off 3pt on 8pt] (s2b) edge node  {} (sna);
            \path[->] (sna) edge node  {} (snb);
            \path[->] (snb) edge [loop right] node  {} (snb);
            
            \path[->] (s1a) edge[bend right=35] node  {} (s0);
            \path[->] (s1b) edge[bend right=35] node  {} (s0);
            \path[->] (s2a) edge[bend right=30] node  {} (s1a);
            \path[->] (s2b) edge[bend right=30] node  {} (s1a);
            \path[->,dashed] (s3a) edge[bend right=30] node  {} (s2a);
            \path[->,dashed] (s3b) edge[bend right=30] node  {} (s2a);
            \path[->] (sna) edge[bend right=30] node  {} (snminus1a);
            \path[->] (snb) edge[bend right=30] node  {} (snminus1a);
            
            \path[->] (s0) edge[bend right=7] node  {} (sf);
            \path[->] (sf) edge[bend right=10] node  {} (sna);
        \end{tikzpicture}
        \caption{Transition system with Büchi objective $\generally \eventually \{s_f\}$ that has exponentially many minimal winning coalitions.}
        \label{fig:buechi_exponential}
    \end{figure}
    
    While Algorithm~\ref{alg:opt_pos_buechi} efficiently checks whether \emph{a} switching pair exists for a Büchi objective, computing responsibility values requires finding \emph{all} switching pairs. This is computationally hard, as a system may have exponentially many minimal winning coalitions, as shown in Figure~\ref{fig:buechi_exponential}: Every $C \in \{\sia{1}, \sib{1}\} \times \cdots \times \{\sia{n}, \sib{n}\}$ is a minimal winning coalition.
    
}

\begin{proposition}
	\label{prop:compl_reach_pes_thresh}
    \label{prop:compl_buechi_opt_thresh}
    \label{prop:compl_buechi_pes_thresh}
    \label{prop:compl_parity_opt_thresh}
    \label{prop:compl_parity_pes_thresh}
    For reachability objectives, the pessimistic threshold problem is NP-hard and in \textsc{PSpace}. For Büchi and parity objectives, the optimistic and pessimistic threshold problems are in \textsc{PSpace} and the pessimistic threshold problem is NP-hard.
\end{proposition}

\begin{appendixproof}
	To decide the corresponding threshold problem for some state $s \in S$ and threshold $t \in [0, 1]$, a Turing machine iterates every subset of states $C \subseteq S$ and determines whether $(C, s)$ forms a switching pair. This takes polynomial space, as reachability, Büchi and parity games can be solved in polynomial space. As these checks can be performed consecutively and only polynomial space is required for keeping track of the responsibility value, computing the responsibility value is possible in polynomial space. The final value is then compared to the threshold~$t$.
    
    For the pessimistic threshold problem and for reachability, Büchi and parity objectives, NP-hardness follows from the NP-hardness of the pessimistic positivity problem for the corresponding objective class (\Cref{prop:compl_reach_pes_pos}).
\end{appendixproof}

\begin{propositionrep}
    \label{prop:compl_reach_pes_comp}
    \label{prop:compl_buechi_opt_comp}
    \label{prop:compl_buechi_pes_comp}
    \label{prop:compl_parity_opt_comp}
    \label{prop:compl_parity_pes_comp}
    For Büchi and parity objectives, the optimistic computation problem is in \sharpp{} and for reachability, Büchi and parity objectives, the pessimistic computation problem is \sharpp{}-complete.
\end{propositionrep}

\nocite{ParityGamesUniqueAcceptingRun}
\begin{appendixproof}
    \newcommand{\snotf}{s_{\lnot F}}
    To show inclusion in \sharpp{}, we construct a non-deterministic Turing machine whose accepting runs correspond to switching pairs. For this, the machine non-de\-ter\-mi\-nis\-ti\-cally chooses a set of states $C$. It then verifies that $(C, s)$ is a switching pair by solving two games, which takes polynomial time. In the case of parity games, no deterministic algorithm is known for solving them in polynomial time, so we rely on the fact that they are in $\text{NP}\cap\text{CoNP}$ to ensure that there is a one-to-one correspondence \cite{ParityGamesUniqueAcceptingRun}.
    
    We first show the NP-hardness of the pessimistic positivity problem for reachability objectives.
    In \cite{ImportanceValuesTL}, it is shown in Theorem IV.6 that the forward responsibility computation problem is \sharpp{}-hard by reduction from the problem of counting the number of solutions to a 1-in-3SAT instance. 
    
    Now consider the class of transition systems that have initial state $\tsInit' \in \tsStates$ with exactly two outgoing transitions to states $\tsInit \in \tsStates$ and $\snotf \in  \tsStates \setminus F$, such that $\snotf$ only has a self loop. The the forward responsibility computation problem restricted to this class is still \sharpp{}-complete, as one can simply add states $\tsInit'$ and $\snotf$ to the construction from \cite{ImportanceValuesTL} without changing the number of accepting runs -- in the modified transition system, every run starts from $\tsInit'$ to $\tsInit$ and then continues as in the original transition system.
    
    Finally, we provide a reduction from the forward responsibility computation problem restricted to the class described above to the backward pessimistic computation problem. Given an instance of the restricted forward responsibility computation problem $\ts, F, s$ with the special states $\tsInit'$ and $\snotf$, we set $\tsRun = \tsInit' \snotf^\omega$. This run violates $F$ as we assumed $\snotf \notin F$. Therefore, we output the backward responsibility computation problem $\ts, \tsRun, F, s$.
    
    We now show that for $C \subseteq \tsStates$, the pair $(C, s)$ is switching in the forward case if and only if it is switching in the backward case. Let $(C, s)$ be a switching pair in the forward case. Then $\tsInit' \in C$, as otherwise Player~$\ggPTwo$ controls $\tsInit'$ and can go to $\snotf$, from where $F$ is unreachable. Therefore, in the backwards case, the transition from $\tsInit'$ to $\tsInit$ is not removed during engraving and Player~$\ggPOne$ controls $\tsInit'$. Due to this, Player~$\ggPOne$ can first go to $\tsInit$ and then follow the strategy from the forward case, which eventually reaches $F$.
    
    The argument in the other direction is analogous: Let $(C, s)$ be a switching pair in the backward setting. Then $\tsInit' \in C$ as otherwise, the engraving would remove the transition from $\tsInit'$ to $\tsInit$, thus making $F$ unreachable. Player~$\ggPOne$ must go from $\tsInit'$ to $\tsInit$ first (because $F$ is unreachable from $\snotf$). The winning strategy it follows from $\tsInit$ is also winning in the forward case, as the games are identical: The engraving does not affect the states that are not on $\tsRun$.
    
    Therefore, the pessimistic backward computation problem for reachability objectives in \sharpp{}-hard.
    
    \paragraph*{}
    To show \sharpp{}-hardness of the pessimistic computation problem for Büchi objectives, we give a reduction from the pessimistic backward responsibility computation problem for reachability objectives. While it is quite easy to model a reachability property as a Büchi property, it is a bit more involved to ensure that the modified system has the same number of runs.
    
    Let $\ts = \tsStates$ be a transition system with reachability property $F$ and run $\tsRun$ violating $F$. We construct $\ts' = (\tsStates \cup \dot{\tsStates}, \tsTrans', \tsInit)$, where $\dot{\tsStates} = \{ \dot{s}\mid s \in \tsStates\}$ and where the transition relation is given by
    \begin{itemize}
        \item $s \tsTransRel' t$ for $s \in \tsStates \setminus F$, $t \in \tsStates$ with $s \tsTransRel t$,
        \item $s \tsTransRel' \dot{s}$ for $s \in F$ and
        \item $\dot{s} \tsTransRel' \dot{t}$ for $s, t \in \tsStates$ with $s \tsTransRel t$.
    \end{itemize}
    As Büchi objective, choose $F' = \dot{\tsStates}$.
    
    We construct a bijection $f \colon \tsRuns(\ts) \to \tsRuns(\ts')$ and show that a run $\tsRun \in \tsRuns(\ts)$ is accepting if and only if $f(\tsRun) \in \tsRuns(\ts')$ is accepting. Let $\tsRun = \tsRun_0 \tsRun_1 \in \tsRuns(\ts)$ be a run. If $\tsRun_i \notin F$ for all $i \in \mathbb N$, then set $f(\tsRun) = \tsRun$. In this case, neither $\tsRun$ nor $f(\tsRun)$ is accepting with regard to reachability objective $F$ and Büchi objective $F'$, respectively. Now consider the case that $\tsRun$ is accepting with regard to reachability objective $F$. Pick the smallest $i \in \mathbb N$ with $\tsRun_i \in F$ and set $f(\tsRun) = \tsRun_0 \cdots \tsRun_i \dot{\tsRun}_i \dot{\tsRun}_{i+1} \cdots$. Then $f(\tsRun)$ is also accepting because for every $j \geq i$, $f(\tsRun)_j \in F'$ and therefore, $F'$ is visited infinitely often.
    
    Due to this bijection that preserves acceptance, the number of accepted runs is the same in $\ts$ and $f(\ts)$. Therefore, the reduction target is also \sharpp{}-hard.
    
    For parity objectives, \sharpp{}-hardness follows from the \sharpp{}-hardness of the pessimistic computation problem for Büchi objectives. Similar to the proof of Proposition~\ref{prop:compl_parity_opt_thresh}, it is sufficient to colour states in $F$ with colour $2$ and those not in $F$ with colour $1$.
\end{appendixproof}
\section{Computing responsible states using refinement}
\label{sec:refinement}

Iterating all switching pairs to compute responsibility is infeasible in large models. However, many models only have a small set of responsible states, with most states bearing no responsibility. Existing implementations use some local criteria to identify states with no responsibility (e.g. states with only a single successor), but this local view does not capture the higher-level model structure. For example, the clouds in Figure~\ref{fig:refinement_sketch} represent a large number of states connected acyclically and have no responsibility for whether $s_+$ is eventually reached, even though the clouds might contain complex local structures.
Existing approaches for dealing with large models include stochastic sampling \cite{BackwardRespTS} and
grouping states into blocks \cite{BackwardRespTS,ActorBasedResponsibility}. The latter approach is efficient as long as the number of state blocks is limited, but requires manual specification of the state blocks.

\begin{figure}[t]
	\centering
	\begin{tikzpicture}[shorten >=1pt,node distance=1.5cm,on grid,auto, state/.style={circle,inner sep=1pt}] 
		\node[state,initial,initial text=] (s0)   {$s_0$};
		\node[cloud, draw,cloud puffs=9,cloud puff arc=160, aspect=1.5, inner ysep=0.35em, right=1.3cm of s0] (c1) {$c_1$};
		\node[state] (s1) [right=1.35cm of c1] {$s_\mathit{crit}$};
		\node[cloud, draw,cloud puffs=9,cloud puff arc=160, aspect=1.5, inner ysep=0.35em, above right=0.4cm and 1.5cm of s1] (c2) {$c_2$};
		\node[cloud, draw,cloud puffs=9,cloud puff arc=160, aspect=1.5, inner ysep=0.35em, below right=0.4cm and 1.5cm of s1] (c3) {$c_3$};
		\node[state] (s2) [right=1.3cm of c2] {$s_+$};
		\node[state] (s3) [right=1.3cm of c3] {$s_-$};
		
		\path[->] (s0) edge node {} (c1);
		\path[->] (c1) edge node {} (s1);
		\path[->] (s1) edge node {} (c2);
		\path[->] (s1) edge node {} (c3);
		\path[->] (s1) edge node {} (c3);
		\path[->] (c2) edge node {} (s2);
		\path[->] (c3) edge node {} (s3);
		\path[->] (s2) edge[loop right] node {} (s2);
		\path[->] (s3) edge[loop right] node {} (s3 );
	\end{tikzpicture}
	\caption{Sketch of a large transition system with reachability property $\lozenge s_+$.
		If the clouds are acyclic, only $s_\mathit{crit}$ has positive responsibility. The refinement algorithm from Section~\ref{sec:refinement} can efficiently identify $s_\mathit{crit}$ as responsible.}
\label{fig:refinement_sketch}
\end{figure}

This section presents an algorithm that automatically builds state blocks by iterative refinement.
To output individual responsibility, we give  a condition under which a block's responsibility corresponds to a state's individual responsibility. We also provide a heuristics that is optimal for safety and reachability objectives and illustrate the algorithm by an example.

Let $\ts = \tsTuple$ be a transition system with objective $\tsProp$ and run $\tsRun$ that violates $\tsProp$, let $\respName \in \{ \respOpt{\ts}{\tsProp}{\tsRun}, \respPes{\ts}{\tsProp}{\tsRun} \}$ be a responsibility metric and, for $C \subseteq \tsStates$, and let $\ggGame[C] \in \{\optGG[\ts, \tsProp, \tsRun, C], \pesGG[\ts, \tsProp, \tsRun, C] \}$ be the corresponding graph game.

\begin{definition}[Block responsibility]
	Let $\stateGroups$ be a partition of $\tsStates$. Every $B \in \stateGroups$ is called a \emph{state block}. For state $s \in \tsStates$, $\groupOf{s}$ denotes the unique block in $\stateGroups$ that contains $s$ and for $S' \subseteq \tsStates$, $\groupOf{S'} = \bigcup_{s \in S'} \groupOf{s}$.
	The \emph{block responsibility} $\respGroupedName$ is defined by the Shapley value of the function $2^\stateGroups \to \{0, 1\}$ that maps $\mathcal{C} \subseteq \stateGroups$ to $\ggValue(\ggGame[\flatten(\mathcal{C})])$, where $\flatten(\mathcal{C}) = \bigcup_{B \in C} B$. If $B \in \stateGroups$ is a block and $C$ is a union of blocks in $\stateGroups$, then the pair $(C, B)$ is called a \emph{block-switching pair} if Player~$\ggPOne$ wins $\ggGame[C \cup B]$ and Player~$\ggPTwo$ wins $\ggGame[C]$. The set of blocks $B$ for which some block-switching pair $(C, B)$ exists is denoted by $\hasSP$.
\end{definition}

\begin{remark}
    A block $B \in \stateGroups$ has positive block responsibility if and only if $B \in \hasSP$, but the responsibility value of a block is not equal to the sum of responsibility values of its members: Figure~\ref{fig:groups_no_relation} demonstrates that no obvious relation exists between block responsibility and (individual) responsibility. If states are partitioned into $\{s_0, s_1\}$ and $s_2$, then block $\{s_0, s_1\}$ has responsibility value $0$. On the other hand, if states are considered individually, then $s_0$ has responsibility value $1/2$ and $s_1$ and $s_2$ have responsibility value $1/4$ each.
    
    This shows that the responsibility value of a block is not the sum of responsibilities of its members and a block with no responsibility may contain members with positive responsibility. Additionally, if some block $B$ is refined, this affects the responsibility values of other blocks. In particular, a block $C$ that had no responsibility before may have a positive responsibility value after block $B$ is refined into $B'$ and $B \setminus B'$.
\end{remark}

\begin{figure}[t]
 \centering
 \begin{tikzpicture}[shorten >=1pt,node distance=1.2cm,on grid,auto, state/.style={circle,inner sep=1pt}] 
      \node[state,initial,initial text=] (s0)   {$s_0$};
      \node[state] (s1) [above right=0.5cm and 1.2 cm of s0] {$s_1$};
      \node[state] (s2) [below right=0.5cm and 1.2 cm of s0] {$s_2$};
      \node[state,fstate] (s3) [right=2.4cm of s0] {$s_3$};
      
      \path (s0) edge[loop above, counterexample] node  {} (s0);
      
      \path[->] (s0) edge node  {} (s1);
      \path[->] (s0) edge node  {} (s2);
      \path[->] (s1) edge node  {} (s3);
      \path[->] (s2) edge node  {} (s3);
      
      \path[->] (s0) edge [loop above] node  {} (s0);
      \path[->] (s1) edge [loop above] node  {} (s1);
      \path[->] (s2) edge [loop below] node  {} (s2);
      \path[->] (s3) edge [loop right] node  {} (s3);
  \end{tikzpicture}
 \caption{A transition system with reachability objective $\eventually \{s_3\}$ and run $\tsRun = s_0^\omega$ violating $F$, demonstrating how block responsibility compares to individual responsibility.}
 \label{fig:groups_no_relation}
\end{figure}

Under the following restriction, a (singleton) block's responsibility is equal to its member state. This is used by the refinement algorithm as a termination condition.

\begin{theorem}
	\label{thm:singleton_sp_bsp}
	If every block in $\hasSP$ is a singleton, then for $s \in \tsStates$,
	\[\respName(s) > 0 \quad \mathit{iff} \quad \respGroupedName(\groupOf{s}) > 0.\]
\end{theorem}

\begin{proof}
	By contradiction, let $s \in \tsStates$ such that $\respName(s) > 0$, but $\respGroupedName(\groupOf{s}) = 0$. As $\respName(s) > 0$, there is some $C \subseteq \tsStates$ such that $(C, s)$ is a switching pair. Conversely, $(\groupOf{C}, \groupOf{s})$ is not a block-switching pair because $\respGroupedName(s) = 0$. As $C \cup \{s\} \subseteq \groupOf{C} \cup \groupOf{s}$, Player~$\ggPOne$ wins $\ggGame[\groupOf{C} \cup \groupOf{s}]$. Thus, Player~$\ggPOne$ also wins $\ggGame[\groupOf{C}]$, as otherwise $(\groupOf{C}, \groupOf{s})$ would be a block-switching pair.
	Let $C' = \{s \in C \mid \groupOf{s} \in \hasSP\}$. As Player~$\ggPOne$ wins $\ggGame[\groupOf{C}]$, Player~$\ggPOne$ also wins $\ggGame[\groupOf{C'}]$ -- otherwise, one could construct a switching pair for some $\groupOf{s'}$ with $s \in C \setminus C'$, but the blocks of states in $C \setminus C'$ do not have switching pairs. As every block in $\hasSP$ is singleton, $\groupOf{C'} = C' \subseteq C$ and thus Player~$\ggPOne$ wins $\ggGame[C]$. This contradicts $(C, s)$ being a switching pair.
	Now let $\respGroupedName(\groupOf{s}) > 0$. Then there exists a block-switching pair $(\groupOf{C}, \groupOf{s})$. Then $\groupOf{s} = \{ s \}$ is singleton and thus $(\groupOf{C}, s)$ is also a switching pair, implying $\respName(s) > 0$.
\end{proof}

Theorem~\ref{thm:singleton_sp_bsp} forms the heart of the responsibility refinement algorithm (Algorithm~\ref{alg:refinement}). Heuristics $\Call{InitialPartition}{\relax}$ produces a coarse initial partition. Each iteration first determines whether any blocks in $\hasSP$ are not singleton and selects some of these for refinement using heuristics $\Call{SelectRefinements}{\cdot}$. Each chosen block is then refined into two or more blocks by the heuristics $\Call{Refine}{\cdot}$. Different heuristics are evaluated in Section~\ref{sec:experimental}.

\begin{algorithm}[t]
	\caption{Refinement algorithm for computing the set of responsible states.}
	\label{alg:refinement}
	\begin{algorithmic}[1]
		\Let{$\stateGroups$}{\Call{InitialPartition}{\relax}}
		\Let{$\hasSP$}{\Call{ComputeHasSP}{$\stateGroups$}}
		\While{$\exists B \in \hasSP \textit{ with } |B| > 1$}
		\Let{$R$}{\Call{SelectRefinements}{$\stateGroups$, $\hasSP$}}
		\For{$B \in R$}
		\Let{$\stateGroups$}{$(\stateGroups \setminus B) \cup \Call{Refine}{B}$}
		\EndFor
		\Let{$\hasSP$}{\Call{ComputeHasSP}{$\stateGroups$}}
		\EndWhile
		\State\Return $\{s \mid \groupOf{s} \in \hasSP\}$
	\end{algorithmic}
\end{algorithm}

\begin{theorem}[Correctness of refinement]
	If $\textsc{SelectRefinements}(\cdot)$ always selects at least one block and if $\textsc{Refine}(\cdot)$ always produces at least two non-empty blocks, then Algorithm~\ref{alg:refinement} terminates and a state $s$ is contained in the algorithm's output if and only if $\respName(s) > 0$.
\end{theorem}

\begin{proof}
	During every iteration of the loop (lines~3--7), the number of (non-empty) blocks in $\stateGroups$ grows, because at least one block is refined and this refinement splits the block into at least two new blocks. Therefore, after finitely many iterations, the partition only contains singletons, unless the algorithm has already terminated. Once the partition only contains singletons, the loop condition (line~3) is no longer satisfied and the algorithm therefore terminates.
	
	Line~8 is reached when all blocks in $\hasSP$ are singletons, so Theorem~\ref{thm:singleton_sp_bsp} is applicable. A block $\groupOf{s}$ has positive group responsibility if and only if a block-switching pair for $\groupOf{s}$ exists and by Theorem~\ref{thm:singleton_sp_bsp}, $s$ has positive (individual) responsibility exactly in this case. Therefore, Algorithm~\ref{alg:refinement} outputs exactly the states with positive individual responsibility.
\end{proof}

The performance of Algorithm~\ref{alg:refinement} depends on suitable heuristics. Section~\ref{sec:experimental} presents and evaluates several heuristics for $\Call{InitialPartition}{\cdot}$, $\Call{SelectRefinements}{\cdot}$, $\Call{Refine}{\cdot}$. In the following, we present an instance of the $\Call{Refine}{\cdot}$ heuristics called \emph{Frontier} and show that it is optimal for safety and reachability objectives.

\newcommand{\Frontier}{\mathit{Fr}}

\begin{definition}[Frontier refinement heuristics]
	\label{def:frontier}
    Let $(B, C)$ be a block-switching pair. Let $\Delta = \ggWin(\ggGame[C \cup B]) \setminus \ggWin(\ggGame[C])$ denote the states that are in Player $\ggPOne$'s winning region if he controls $C \cup B$ but not in his winning region if he only controls $C$.
    
    The \emph{frontier $\Frontier(B, C)$} contains the states $s \in \Delta \cap B$ such that $s$ has a transition to some state $t \notin \Delta$. The \emph{frontier refinement heuristics} selects a state $s \in \Frontier(B, C)$ and refines block $B$ into $B \setminus \{s\}$ and $\{s\}$. If $\Frontier = \varnothing$ exists, the heuristics selects an arbitrary state from $\Delta$.
\end{definition}

\begin{lemma}
    \label{lem:frontier_well_defined}
    Let $(B, C)$ be a block-switching pair and $\Delta = \ggWin(\ggGame[C \cup B]) \setminus \ggWin(\ggGame[C])$. For safety, reachability, Büchi and parity objectives, $\Delta \cap B \neq \varnothing$ and therefore, the frontier refinement heuristics can always select a state $s \in \Delta$.
\end{lemma}

\begin{proof}
    As $(B, C)$ is a block-switching pair, Player $\ggPOne$ wins $\ggGame[C \cup B]$ and thus the initial state $\tsInit \in \ggWin(\ggGame[C \cup B])$. Similarly, Player $\ggPOne$ loses $\ggGame[C]$ and thus $\tsInit \notin \ggWin(\ggGame[C])$. Therefore, $\ggWin(\ggGame[C \cup B]) \setminus \ggWin(\ggGame[C]) \neq \varnothing$.
    
    To show that $\Delta \cap B \neq \varnothing$, the proof uses the fact that, for safety, reachability Büchi and parity objectives and for any $\tsStates' \subseteq \tsStates$, Player $\ggPOne$ can win from the states in $\ggWin(\ggGame[\tsStates'])$ without leaving $\ggWin(\ggGame[\tsStates'])$ and Player $\ggPTwo$ can win from the states in $\tsStates \setminus \ggWin(\ggGame[\tsStates'])$ without leaving $\tsStates \setminus \ggWin(\ggGame[\tsStates'])$.
    
    Therefore, Player $\ggPOne$ does not need to control any of the states in $B$ that are outside of $\ggWin(\ggGame[C \cup B])$, so he wins $\ggGame[C \cup (B \cap \ggWin(\ggGame[C \cup B]))]$. By the same argument, once a play reaches $\ggWin(\ggGame[C])$, controlling the states in $C$ is sufficient for remaining in $\ggWin(\ggGame[C])$, so Player $\ggPOne$ also wins $\ggGame[C \cup (B \cap \ggWin(\ggGame[C \cup B]) \setminus \ggWin(\ggGame[C]))] = \ggGame[C \cup (B \cap \Delta)]$. If $B \cap \Delta = \varnothing$, then $\ggGame[C \cup (B \cap \Delta)] = \ggGame[C]$, which is a contradiction, as Player $\ggPOne$ wins the former and loses the latter game.
\end{proof}

\begin{figure}[t]
	\centering
	\begin{tikzpicture}[shorten >=1pt,node distance=1.1cm,on grid,auto, state/.style={circle,inner sep=1pt}] 
		\node[state,initial,initial text=] (s0)   {$s_0$};
		\node[state] (s1) [right=of s0] {$s_1$};
		\node[state] (s2) [right=of s1] {$s_2$};
		\node[state] (s3) [right=of s2] {$s_3$};
		\node[state] (s5) [right=of s3] {$s_5$};
		\node[state] (s4) [below=0.7cm of s5] {$s_4$};
		\node[state] (s6) [right=of s5] {$s_6$};
		\node[state] (s7) [right=of s4] {$s_7$};
		\node[state] (s8) [right=of s6] {$s_8$};
		\node[state] (s9) [right=of s8] {$s_9$};
		\node[state] (sbad) [right=of s9] {$\lightning$};
		
		\path[->] (s0) edge node  {} (s1);
		\path[->] (s0) edge[bend left=35] node  {} (s2);
		
		\path[->] (s1) edge[bend right=15] node  {} (s2);
		
		\path[->] (s2) edge[bend right=15] node  {} (s1);
		\path[->] (s2) edge node  {} (s3);
		
		\path[->] (s3) edge node  {} (s4);
		\path[->] (s3) edge node  {} (s5);
		\path[->] (s3) edge [bend left=35] node  {} (s6);
		
		\path[->] (s4) edge[bend right=15] node  {} (s7);

		\path[->] (s5) edge node  {} (s6);
		
		\path[->] (s6) edge node  {} (s8);
		\path[->] (s6) edge [bend left=35] node  {} (s9);
		
		\path[->] (s7) edge[bend right=15] node  {} (s4);
		
		\path[->] (s8) edge node  {} (s7);
		\path[->] (s8) edge node  {} (s9);
		
		\path[->] (s9) edge node  {} (sbad);
	\end{tikzpicture}
	\caption{A transition system with safety objective $\generally \lnot \lightning$. In \Cref{example:refinement}, the refinement algorithm is executed on this system.}
	\label{fig:refinement_example}
\end{figure}

\begin{example}
	\label{example:refinement}
	Consider the transition system in \Cref{fig:refinement_sketch} with safety objective $\generally \lnot \lightning$. In the following, the refinement algorithm is executed with the following heuristics: $\Call{InitialPartition}{\relax}$ constructs a singleton partition. $\Call{SelectRefinements}{\cdot}$ selects the largest non-singleton block. $\Call{Refine}{\cdot}$ uses the frontier refinement heuristics and selects the block with the lowest index from the frontier.
	
	\begin{enumerate}
		\item The initial partition is the singleton partition $\{\{s_0, \ldots, s_9, \lightning\}\}$. This yields the block-switching pair $(\{s_0, \ldots, s_9, \lightning\}, \varnothing)$, which is selected for refinement. The associated winning regions are $\ggWin(\ggGame[\varnothing]) = \{s_4, s_7\}$, and $\ggWin(\ggGame[\{s_0, \ldots, s_9, \lightning\}]) = \{s_0, \ldots, s_8\}$. Therefore, the frontier is $\Frontier(\{s_0, \ldots, s_9, \lightning\}, \varnothing) = \{s_3, s_6, s_8\}$ and state $s_3$ is split from the rest of its block.
		
		\item The new partition is $\{\{s_0, s_1, s_2, s_4, \ldots, s_9, \lightning\}, \{s_3\}\}$. This yields block-switching pairs $(\{s_0, s_1, s_2, s_4, \ldots, s_9, \lightning\}, \varnothing)$ and $(\{s_3\}, \varnothing)$. The former block-switching pair is selected for refinement. The associated winning regions are $\ggWin(\ggGame[\varnothing]) = \{s_4, s_7\}$ and $\ggWin(\ggGame[\{s_0, s_1, s_2, s_4, \ldots, s_9, \lightning\}]) = \{ s_0, \ldots, s_8 \}$. Therefore, the frontier is $\Frontier(\{s_0, s_1, s_2, s_4, \ldots, s_9, \lightning\}, \varnothing) = \{s_6, s_8\}$ and state $s_6$ is split from the rest of its block.
		
		\item The new partition is $\{\{s_0, s_1, s_2, s_4, s_5, s_7, s_8, s_9, \lightning\}, \{s_3\}, \{s_6\}\}$. This yields block-switching pairs $(\{s_0, s_1, s_2, s_4, s_5, s_7, s_8, s_9, \lightning\},\varnothing)$ (which is selected for refinement), $(\{s_3\}, \varnothing)$ and $(\{s_6\}, \{s_0, s_1, s_2, s_4, s_5, s_7, s_8, s_9, \lightning\})$. The associated winning regions are $\ggWin(\ggGame[\varnothing]) = \{s_4, s_7\}$ and $\ggWin(\ggGame[\{s_0, s_1, s_2, s_4, s_5, s_7, s_8, s_9, \lightning\}])=$ $\{s_0,$ $s_1, s_2, s_4, s_7\}$. Therefore, the frontier is $\Frontier(\{s_0, s_1, s_2, s_4, s_5, s_7, s_8, s_9, \lightning\},\varnothing = \{s_2\}$ and state $s_2$ is split from the rest of its block.
		
		\item The new partition is $\{\{s_0, s_1, s_4, s_5, s_7, s_8, s_9, \lightning\}, \{s_3\}, \{s_6\}, \{s_2\}\}$. This yields block-switching pairs $(\{s_0, s_1, s_4, s_5, s_7, s_8, s_9, \lightning\}, \{s_6\})$ (which is selected for refinement), $(\{s_3\}, \varnothing)$, $(\{s_6\}, \{s_0, s_1, s_4, s_5, s_7, s_8, s_9, \lightning\})$ and $(\{s_2\}, \varnothing)$. The associated winning regions are $\ggWin(\ggGame[\{s_6\}]) = \{s_4, s_7\}$ and $\ggWin(\ggGame[\{s_0, s_1, s_4, s_5, s_6, s_7, s_8, s_9, \lightning\}]) = \{s_0, \ldots, s_8\}$. Therefore, the frontier is $\Frontier(\{s_0, s_1, s_4, s_5, s_7, s_8, s_9, \lightning\}, \{s_6\}) = \{s_8\}$ and state $s_8$ is split from the rest of the block.
		
		\item The new partition is $\{\{s_0, s_1, s_4, s_5, s_7, s_9, \lightning\}, \{s_3\}, \{s_6\}, \{s_2\}, \{s_8\}\}$. This yields block-switching pairs $(\{s_3\}, \varnothing)$, $(\{s_6\}, \{s_8\})$, $(\{s_2\}, \varnothing)$ and $(\{s_8\}, \{s_6\})$. As every block-switching pair $(B, C)$ has a singleton block $B$, the algorithm terminates.
	\end{enumerate}
	
	To compute responsibility values, only states $s_2$, $s_3$, $s_6$ and $s_8$ need be considered when iterating potential switching pairs. This significantly reduces runtime. In the example, the final responsibility values are $s_2 \mapsto 5 / 12$, $s_3 \mapsto 5 /12 $, $s_6 \mapsto 1 / 12$ and $s_8 \mapsto 1 / 12$. Using the frontier heuristics therefore identified exactly the states with positive responsibility.
\end{example}

For safety and reachability objectives, the frontier always contains a state with positive responsibility. As a first step towards the proof, the following lemma shows that the ``fallback'' option of choosing an arbitrary state from $\Delta$ is not necessary for these objective classes.

\begin{lemma}
    \label{lem:frontier_no_fallback}
    Let $(B, C)$ be a block-switching pair. For reachability and safety objectives, $\Frontier(B, C) \neq \varnothing$.
\end{lemma}

\begin{proof}
	Let $\Delta = \ggWin(\ggGame[C \cup B]) \setminus \ggWin(\ggGame[C])$.
	We first prove the statement for a safety objective $\generally \lnot F$. Because a block-switching pair exists, $F$ is non-empty and reachable from initial state $\tsInit$. Regardless of which states Player $\ggPOne$ controls, $F$ is not in his winning region. Therefore, $\tsStates \setminus \ggWin(\ggGame[C \cup B]) \neq \varnothing$. Because Player $\ggPTwo$ wins $\ggGame[C]$, there is a path from $\tsInit$ to $F$ that does not visit $\ggWin(\ggGame[C])$. As $\tsInit \in \Delta$ and $F \subseteq \tsStates \setminus \ggWin(\ggGame[C \cup B])$, there is some state in $\Delta$ that has a transition to $\tsStates \setminus \ggWin(\ggGame[C \cup B])$.
    
    Every state in $s \in \Delta$ with such a transition is controlled by Player $\ggPOne$ in $\ggGame[C \cup B]$: Otherwise, Player $\ggPTwo$ could play the transition to $\tsStates \setminus \ggWin(\ggGame[C \cup B])$ and $s$ would itself not be in $\Delta$. For the same reason, every such $s$ has a transition to $\ggWin(\ggGame[C \cup B])$.
    At least one state $s \in \Delta$ with a transition to $\tsStates \setminus \ggWin(\ggGame[C \cup B])$ is not controlled by Player $\ggPOne$ in $\ggGame[C]$. Otherwise, Player $\ggPOne$ could always play a transition from $s$ to $\ggWin(\ggGame[C \cup B])$, making $F$ unreachable from $\tsInit$  in $\ggGame[C]$. This is a contradiction, as Player $\ggPTwo$ wins $\ggGame[C]$.
    Because at least one $s \in \Delta$ has a transition to $\tsStates \setminus \ggWin(\ggGame[C \cup B])$, because all such $s$ are in $C \cup B$ and because at least one such $s$ is not in $B$, it holds that $\Frontier(B, C) \neq \varnothing$.
    
    We now prove the statement for a reachability objective $\eventually F$. Because a block-switching pair exists, $F$ is non-empty and reachable from initial state $\tsInit$. Regardless of which states Player $\ggPOne$ controls, $F$ is in his winning region. Therefore, $\ggWin(\ggGame[C]) \neq \varnothing$. Because Player $\ggPOne$ wins $\ggGame[C \cup B]$, there is a path from $\tsInit$ to $F$ that does not visit $\tsStates \setminus \ggWin(\ggGame[C \cup B])$. As $\tsInit \in \Delta$ and $F \subseteq \ggWin(\ggGame[C])$, there is some state in $\Delta$ that has a transition to $\ggWin(\ggGame[C])$.
    
    No state in $s \in \Delta$ with such a transition is controlled by Player $\ggPOne$ in $\ggGame[C]$. Otherwise, Player $\ggPOne$ could play the transition to $\ggWin(\ggGame[C])$ and $s$ would itself be in $\ggWin(\ggGame[C])$. For the same reason, every such $s$ has a transition to $\tsStates \setminus \ggGame[C]$.
    At least one state $s \in \Delta$ with a transition to $\ggWin(\ggGame[C])$ is controlled by Player $\ggPOne$ in $\ggGame[C \cup B]$. Otherwise, Player $\ggPTwo$ could play the transition to $\tsStates \setminus \ggWin(\ggGame[C])$, making $F$ unreachable from $\tsInit$ in $\ggGame[C \cup B]$. This is a contradiction, as Player $\ggPOne$ wins $\ggGame[C \cup B]$.
    Because at least one $s \in \Delta$ has a transition to $\ggWin(\ggGame[C])$, because no such $s$ are in $B$ and because at least one such $s$ is in $C \cup B$, it holds that $\Frontier(B, C) \neq \varnothing$.
\end{proof}

\begin{proposition}[Usefulness of the frontier heuristics for safety and reachability objectives]
	\label{prop:refinement_frontier_useful}
    Let $(B, C)$ be a block-switching pair. For safety and reachability objectives, at least one state in $\Frontier(B, C)$ has positive responsibility.
\end{proposition}

\begin{proof}
	We first prove the statement for safety objective $\generally \lnot F$. Because a block-switching pair exists, $F$ is non-empty and reachable from initial state $\tsInit$. We first show that Player $\ggPOne$ wins $\ggGame[C \cup \Frontier(B, C)]$ by playing as follows: For states in $C$, follow the strategy used to achieve $\ggWin(\ggGame[C])$. For states in $\Frontier(B, C)$, take any transition to $\ggWin(\ggGame[C \cup B])$. Every state in $\Frontier(B, C)$ has a transition to $\ggWin(\ggGame[C \cup B])$, as otherwise, it would itself not be in $\ggWin(\ggGame[C \cup B])$.
	
	We show that Player $\ggPOne$ wins with this strategy by contradiction. Assume that Player $\ggPTwo$ were to win $\ggGame[C \cup \Frontier(B, C)]$. Then Player $\ggPTwo$ has a strategy that induces a path $\tsRun$ from $\tsInit$ to $F$ without visiting $\ggWin(\ggGame[C])$. Let $\tsRun_i$ be the final state of $\tsRun$ that is in $\ggWin(\ggGame[C \cup B]) \setminus \ggWin(\ggGame[C])$. This exists because $\tsInit \in \ggWin(\ggGame[C \cup B]) \setminus \ggWin(\ggGame[C])$ and $F \subseteq \tsStates \setminus \ggWin(\ggGame[C \cup B])$. Then $\tsRun_i$ is controlled by Player $\ggPOne$ in $\ggGame[C \cup B]$. Otherwise, Player $\ggPTwo$ could play the transition to $\tsStates \setminus \ggWin(\ggGame[C \cup B])$ and $\tsRun_i$ would not be in $\ggWin(\ggGame[C \cup B]) \setminus \ggWin(\ggGame[C])$. Therefore, $\tsRun_i \in \Frontier(B, C)$ and thus Player $\ggPOne$ can prevent Player $\ggPTwo$ from playing $\tsRun$ by playing a different transition in $\tsRun_i$. Such a transition exists because otherwise, $\tsRun_i \notin \ggWin(\ggGame(C \cup B))$.
	
	To obtain a switching pair, observe that Player $\ggPTwo$ wins $\ggGame[C]$ and Player $\ggPOne$ wins $\ggGame[C \cup \Frontier(B, C)]$. Therefore, successively giving Player $\ggPOne$ control of the states in $\Frontier(B, C)$ will eventually allow him to win. More formally, there is some $s \in \Frontier(B, C)$ and $\Frontier' \subseteq \Frontier(B, C)$ such that Player $\ggPTwo$ wins $\ggGame[C \cup \Frontier']$ and Player $\ggPOne$ wins $\ggGame[C \cup \Frontier' \cup \{s\}]$. Then $(s, C \cup \Frontier')$ forms a switching pair and thus $s$ has positive responsibility.
	
    We now prove the statement for a reachability objective $\eventually F$. Because a block-switching pair exists, $F$ is non-empty and reachable from initial state $\tsInit$. For every $s \in \Frontier(B, C)$, there is a transition to $S \setminus \ggWin(\ggGame[C])$, as otherwise, $s$ would be included in $\ggWin(\ggGame[C])$. Therefore, if Player $\ggPTwo$ controls the states in $\Frontier(B, C)$, she can prevent any state in $\ggWin(\ggGame[C \cup B]) \setminus \ggWin(\ggGame[C])$ from reaching $\ggWin(\ggGame[C])$. Because $F \subseteq \ggWin(\ggGame[C])$, this implies that Player $\ggPTwo$ wins $\ggGame[S \setminus \Frontier(B, C)]$. On the other hand, Player $\ggPOne$ wins $\ggGame[S]$ (otherwise, no block-switching pairs would exist). Similar to the safety case, there exists $\Frontier' \subseteq \Frontier(B, C)$ and $s \in \Frontier(B, C)$ such that Player $\ggPTwo$ wins $\ggGame[(S \setminus \Frontier(B, C)) \cup \Frontier']$ and $\ggGame[(S \setminus \Frontier(B, C)) \cup \Frontier' \cup \{s\}]$. Then $(s, (S \setminus \Frontier(B, C)) \cup \Frontier')$ forms a switching pair and thus $s$ has positive responsibility.
\end{proof}

\begin{figure}[t]
	\centering
	\begin{tikzpicture}[shorten >=1pt,node distance=1.2cm,on grid,auto, state/.style={circle,inner sep=1pt}] 
		\node[state,initial,initial text=] (s0)   {$s_0$};
		\node[state] (s1) [right=of s0] {$s_1$};
		\node[state] (s2) [above=of s1] {$s_2$};
		\node[state] (sbad) [right=of s1] {$\lightning$};
		
		\path (s0) edge[counterexample] node  {} (s1);
		\path (s1) edge[counterexample] node  {} (sbad);
		\path (sbad) edge[loop right, counterexample] node  {} (sbad);
		
		\path[->] (s0) edge node  {} (s1);
		
		\path[->] (s1) edge node  {} (s2);
		\path[->] (s1) edge node  {} (sbad);
		
		\path[->] (s2) edge node  {} (sbad);
		
		\path[->] (s1) edge [loop below] node  {} (s1);
		\path[->] (s2) edge [loop left] node  {} (s2);
		\path[->] (sbad) edge [loop right] node  {} (sbad);
	\end{tikzpicture}
	\caption{A transition system with safety objective $\Omega= \generally \lnot \lightning$ and run $\tsRun = s_0 s_1 \lightning^\omega$ violating $\Omega$. For $B = \{s_0, s_1, s_2, \lightning\}$ and $C=\varnothing$, the frontier $\Frontier(B, C)$ contains $s_1$ and $s_2$, but (individual) responsibility of $s_2$ is $0$.}
	\label{fig:frontier_no_responsibility}
\end{figure}

\begin{remark}
	At least one state in the frontier has positive responsibility, but the frontier can contain states with no responsibility. An example is given in \Cref{fig:frontier_no_responsibility}. There, the frontier contains $s_1$ and $s_2$, but only $s_1$ has positive responsibility. This is because $s_2$ can only be reached through $s_1$: If Player $\ggPOne$ controls $s_1$, then controlling $s_2$ is not necessary for winning (because $s_0 s_1^\omega$ is winning for Player $\ggPOne$. If Player $\ggPTwo$ controls $s_1$, then Player $\ggPTwo$ can force $s_0 s_1 \lightning$ and win without visiting $s_2$.
	
	\Cref{sec:experimental} experimentally evaluates different schemes for choosing a state from the frontier.
\end{remark}

\begin{figure}[t]
	\centering
	\begin{tikzpicture}[shorten >=1pt,node distance=1.2cm,on grid,auto, state/.style={circle,inner sep=1pt}] 
		\node[state,initial,initial text=] (s0)   {$s_0$};
		\node[state,fstate] (s1) [right=2.0cm of s0] {\phantom{$s_1$}};
		\node[state] (s2) [above right=0.8cm and 1.0cm of s0] {$s_2$};
		\node[state] (s3) [below right=0.8cm and 1.0cm of s0] {$s_3$};
		
		\path (s0) edge[counterexample] node  {} (s1);
		\path (s1) edge[loop right, counterexample] node  {} (s1);
		\node[state,fstate, fill=white] (s1clone) [right=2.0cm of s0] {$s_1$};
		
		\path[->] (s0) edge node  {} (s1);
		
		\path[->] (s1) edge node  {} (s2);
		\path[->] (s1) edge node  {} (s3);
		
		\path[->] (s2) edge node  {} (s0);
		
		\path[->] (s3) edge node  {} (s0);
		
		\path[->] (s1) edge [loop right] node  {} (s1);
	\end{tikzpicture}
	\caption{A transition system with Büchi objective $\Omega= \generally \eventually s_3$ and run $\tsRun = s_0 s_1^\omega$ violating $\Omega$. The frontier $\Frontier(\{s_0, s_1, s_2, s_3\}, \varnothing)$ is empty.}
	\label{fig:no_delta_for_buechi}
\end{figure}

\begin{remark}
	\Cref{lem:frontier_no_fallback} does not hold for Büchi objectives. Consider the transition system in \Cref{fig:no_delta_for_buechi} with Büchi objective $\Omega = \generally \eventually s_3$. Let $B = \{s_0, s_1, s_2, s_3\}$ and $C = \varnothing$. Then $\ggWin(\ggGame[C]) = \varnothing$ and $\ggWin(\ggGame[C \cup B]) = \{s_0, s_1, s_2, s_3\}$. There is no transition from a state in $\ggWin(\ggGame[C \cup B]) \setminus \ggWin(\ggGame[C])$ that leaves this set. Therefore, $\Frontier(B, C) = \varnothing$.
\end{remark}
\section{Implementation and evaluation}

We have developed a unified tool to compute responsibility values. The tool is available as a reproducible artifact \cite{Artifact} and the current version is hosted on Github\footnote{\url{https://github.com/johannesalehmann/SVaBResp}}. The tool computes responsibility according to the following pipeline:

\begin{enumerate}
    \item \textbf{Model building:} The model is given as a \textsc{Prism} \cite{PrismModelChecker} file. The model must not contain any probabilistic choices. It is built according to the standard semantics of the \textsc{Prism} language. This no longer requires external tools, thereby making the tool much easier to install and distribute.
    \item \textbf{Pre-processing:} Several heuristics are used to identify states that are guaranteed to have no responsibility. This includes
    \begin{itemize}
        \item states with only one outgoing transition and
        \item states that are in the winning region of Player~$\ggPOne$ when Player~$\ggPTwo$ controls all states or that are not in the winning region of Player~$\ggPOne$ when Player~$\ggPOne$ controls all states.
    \end{itemize}
    Previous implementations do not provide the second heuristics, which leads to significant pruning in many models.
    \item \textbf{State grouping} Optionally, a state grouping scheme is applied to the model, as described in \cite{ActorBasedResponsibility}. State grouping may additionally require pre-processing the \textsc{Prism} input. The model, objective and partitioning of the state space then form a cooperative game.
    \item \textbf{Responsibility computation} The responsibility of the players in the cooperative game is then analysed, either by directly enumerating all coalitions or by first performing refinement to identify the set of responsible states.
\end{enumerate}

The tool is highly modular and individual stages can be replaced, extended or used independently of the rest of the tool.

\subsection{Visualisation}

\begin{figure}
	\centering
	\includegraphics[width=0.95\textwidth]{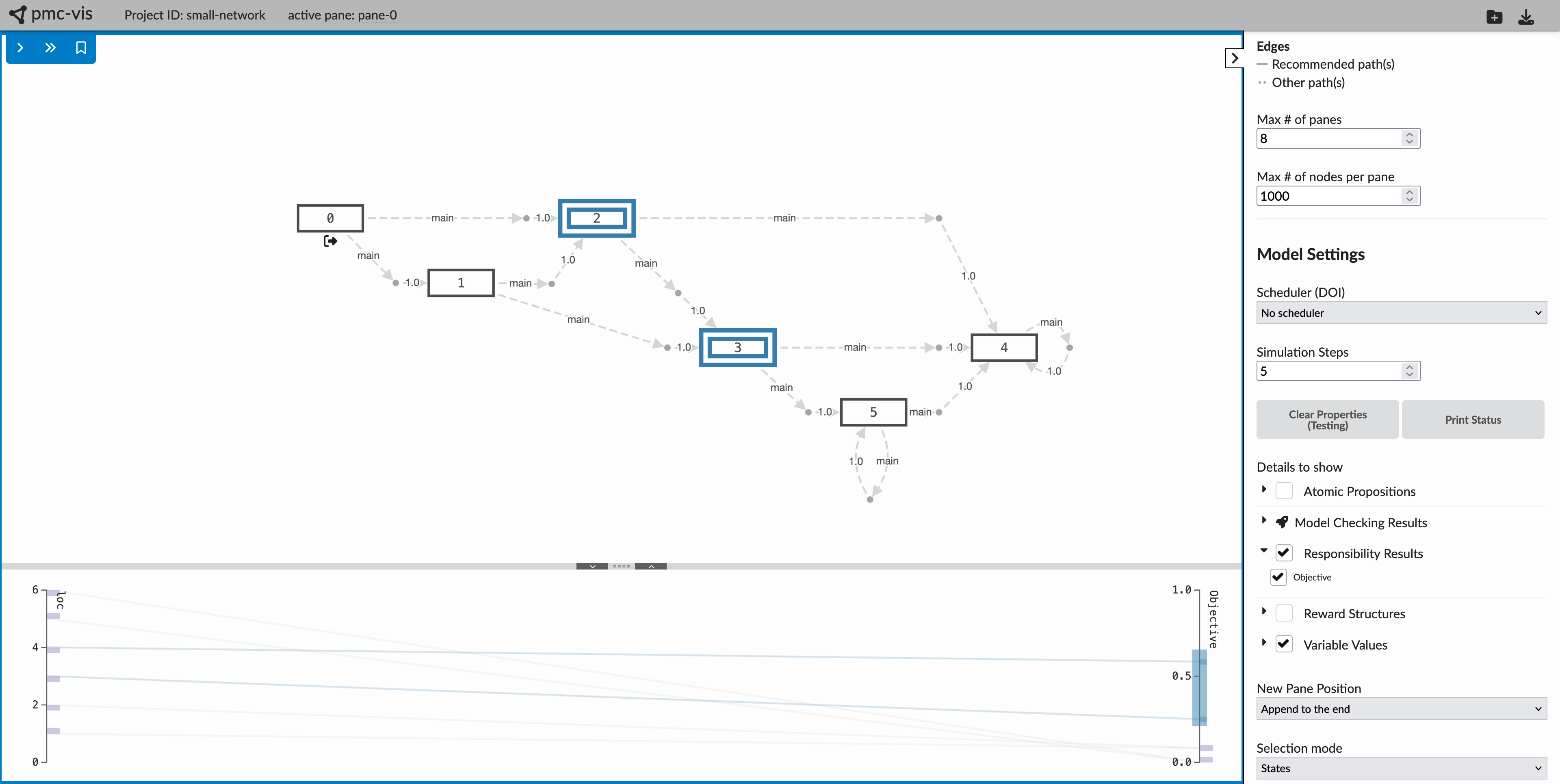}
	\caption{A screenshot of PMC-Vis displaying responsibility values. The main view shows the model as a graph structure, the right panel allows selecting which state properties to visualise (in this case, variable values and responsibility values are selected) and the bottom panel displays the selected properties. The model has a single variable \texttt{loc} and a single objective, for which the responsibility values are displayed on the right of the bottom panel. The lines in the bottom panel correspond to the states. For example, in the state with the highest responsibility value, \texttt{loc} has value $4$. The bottom view allows filtering states, which affects the main view: In the example, the states with the highest responsibility are selected in the bottom view and therefore highlighted in the main view.}
	\label{fig:pmc_vis_pcp}
\end{figure}

Similar to the prototypical implementations of previous papers, the tool is invoked primarily through the command line. This is very portable, but makes it difficult to compare the responsibility results to other properties of the states, such as model checking results and atomic propositions.

To this end, we provide an integration with PMC-VIS \cite{PMCVIS}, an interactive visualisation tool for probabilistic model checking. The integration is available as a Docker container \cite{PMCVisDocker}. PMC-VIS provides a browser-based interface that allows exploring the state space of a model. A \emph{parallel coordinates plot} displays the properties of the states, including the value of model variables, which atomic propositions are enabled, model checking results and -- with the integration developed by us -- responsibility values. The parallel coordinates plot is interactive and states can be filtered based on which values they fulfil. PMC-VIS supports models with multiple properties, yielding multiple responsibility values per state.

In Figure~\ref{fig:pmc_vis_pcp}, all states with $x=0$ and a responsibility of at least $0.2$ are selected. The parallel coordinates plot highlights the other properties of the selected states and in the main view, the states are selected.

Additionally, PMC-VIS provides an interactive \textsc{Prism} editor. In this view, our integration provides buttons to compute per-module and per-action responsibility. Modules and actions with positive responsibility are highlighted in the code view and a tool-tip provides detailed information of responsibility. An example of this is given in Figure~\ref{fig:pmc_vis_code} as part of the case study in \Cref{sec:casestudyone} .

\subsection{Refinement}

\label{sec:experimental}

The refinement algorithm is based on the pseudo-code given in Algorithm~\ref{alg:refinement}. The implementation makes use of additional optimisations, such as ignoring singleton blocks from $\textsc{ComputeHasSP}(\stateGroups)$.

For suitable models, the refinement algorithm reduces responsibility computation time significantly and enables analysis of models that previously timed out. All tests were run on an M2 MacBook Pro with 24 GB or RAM. Each result is the average of 10 runs.

\paragraph{Models.} The heuristics are evaluated on the following models.

\begin{itemize}
	\item \texttt{generals} models an expanded version of the two generals problem \cite{TwoGenerals}, with the objective of either all generals attacking or none of them attacking.
	\item \texttt{station} models a large railway station, where trains must avoid certain tracks.
	\item \texttt{philosophers} models the dining philosophers problem. As every player has positive responsibility, the refinement algorithm cannot improve performance here. However, the performance cost of running refinement before the brute-force algorithm is small, as Table~\ref{tab:initial_partition} shows.
	\item \texttt{large\_frontier\_reach} is a stress test that has a large frontier, both with states that can reach $\ggWin(\ggGame[C])$ and states than can reach $\tsStates \setminus \ggWin(\ggGame[C \cup B])$ for block-switching pair $(B, C)$.
	\item \texttt{large\_frontier\_safety} is similar, except with a safety objective.
	\item \texttt{almost\_empty\_frontier} also has a large frontier, but only one state in the frontier has positive responsibility.
	\item \texttt{centrifuge} models a medical lab. It is discussed in more detail in \Cref{sec:casestudyone}.
	\item \texttt{clouds} models the model from Figure~\ref{fig:refinement_sketch}, where each cloud contains $100\;000$ states. Only a single state has positive responsibility.
\end{itemize}

\begin{table}[t]{
		\setlength{\tabcolsep}{8.3pt}
		\caption{Overview of the models used to evaluate refinement heuristics.}
		\label{tab:model_overview}
		{
			\centering		
			\begin{tabular}{llcrrr}
				\toprule
				&&&&\multicolumn{2}{c}{players}\\
				\cmidrule(lr){5-6}
				\multicolumn{1}{c}{name}
				& \multicolumn{1}{c}{objective}
				& \multicolumn{1}{c}{grouping}
				& \multicolumn{1}{c}{states}
				& \multicolumn{1}{c}{total}
				& \multicolumn{1}{c}{resp.}
				\\\midrule
				\texttt{generals} &
				reachability &
				- &
				60 &
				38 &
				20 \\
				\texttt{station} &
				safety &
				- &
				84 &
				27 &
				24 \\
				\texttt{philosophers} &
				safety &
				- &
				36 &
				22 &
				22 \\
				\texttt{large\_frontier\_reach} &
				reachability &
				- &
				41 &
				39 &
				21 \\
				\texttt{large\_frontier\_safety} &
				safety &
				- &
				41 &
				39 &
				21 \\
				\texttt{almost\_empty\_frontier} &
				reachability &
				- &
				101 &
				100 &
				1 \\
				\texttt{centrifuges} &
				reachability &
				modules &
				$1\;715\;201$ &
				53 &
				2 \\
				\texttt{clouds} &
				reachability &
				- &
				$999\;002$ &
				$332\;999$ &
				1 \\
				\texttt{complex\_clouds} &
				Büchi &
				- &
				$999\;999$ &
				$857\;136$ &
				4 \\
				\bottomrule
			\end{tabular}\par
		}
}\end{table}

An overview of the model parameters is given in Table~\ref{tab:model_overview}. Here, \emph{states} indicates the number of states of the model and \emph{total players} denotes the number of players in the corresponding cooperative game. For models without grouping, there is one player for every state unless the state trivially has no responsibility (cf. the pre-processing step described above). For models with grouping, there is one player per group. \emph{Resp. players} indicates how many players have positive responsibility.

\paragraph{Initial partition.}

\begin{table}[t]{
		\setlength{\tabcolsep}{7.5pt}
		\newcommand{\statCell}[2][s\phantom{i}]{\hspace{-0.15em}#2\,#1}
		\newcommand{\unrefinedCell}[2][s\phantom{i}]{#2\,#1}
		\caption{Comparison of runtimes for different initial partition sizes $n$. TO indicates that computation did not finish without 60 seconds.}
		\label{tab:initial_partition}
		{
			\centering		
			\begin{tabular}{lrrrrrr}
				\toprule
				& \multicolumn{1}{c}{\hspace{-3.25em}\emph{{no refinement}}}
				& \multicolumn{1}{c}{\emph{$n=1$}}
				& \multicolumn{1}{c}{\emph{$n=2$}}
				& \multicolumn{1}{c}{\emph{$n=3$}}
				& \multicolumn{1}{c}{\emph{$n=4$}}
				& \multicolumn{1}{c}{\emph{$n=5$}}
				\\\midrule
				\texttt{generals} &
				TO &
				\statCell[ms]{963} &
				\statCell{1.20} &
				\statCell{1.61} &
				\statCell{2.83} &
				\statCell{4.90} \\
				\texttt{station} &
				\statCell{25.93} &
				\statCell{7.65} &
				\statCell{15.73} &
				\statCell{24.48} &
				\statCell{20.22} &
				\statCell{23.35} \\
				\texttt{philosophers} &
				\statCell[ms]{729} &
				\statCell[ms]{913} &
				\statCell[ms]{882} &
				\statCell[ms]{898} &
				\statCell[ms]{891} &
				\statCell[ms]{884} \\
				\texttt{large\_frontier\_reach} &
				TO &
				TO &
				TO &
				TO &
				TO &
				TO \\
				\texttt{large\_frontier\_safety} &
				TO &
				TO &
				TO &
				TO &
				TO &
				TO \\
				\texttt{almost\_empty\_frontier} &
				TO &
				TO &
				TO &
				TO &
				\statCell[ms]{66} &
				\statCell[ms]{13} \\
				\texttt{centrifuges} &
				TO &
				\statCell{23.14} &
				\statCell{23.23} &
				\statCell{23.75} &
				\statCell{24.53} &
				\statCell{26.37} \\
				\texttt{clouds} &
				TO &
				\statCell{1.73} &
				\statCell{1.79} &
				\statCell{1.82} &
				\statCell{1.95} &
				\statCell{2.17} \\
				\texttt{complex\_clouds} &
				TO &
				\statCell{3.01} &
				\statCell{4.07} &
				\statCell{4.83} &
				\statCell{6.36} &
				\statCell{9.15} \\
				\bottomrule
			\end{tabular}
			\par
		}
}\end{table}

To construct the initial partition, the players are randomly assigned to $n$ blocks. For $1 \leq n \leq 5$, the effect this has on runtime is depicted in Table~\ref{tab:initial_partition} (for block selection heuristics \emph{random} and splitting heuristics \emph{frontier-random}).

The largest change is observed in \texttt{almost\_empty\_frontier}, which goes from a timeout to being solved in milliseconds as $n$ increases. In this model, the frontier heuristics is unable to identify which frontier state responsible, so it has to randomly select states until finding the right one. By starting with several blocks, fewer states have to be selected on average before finding the right one. On the other hand, increasing $n$ leads to worse runtime in models where the frontier heuristics is effective, such as \texttt{generals} and \texttt{complex\_clouds}. Here, more initial blocks provide little benefit, while increasing how many coalitions have to be investigated. In models where many states are responsible, the number of initial blocks makes little difference.

\begin{table}[b]{
		\setlength{\tabcolsep}{13.1pt}
		\newcommand{\statCell}[2][s\phantom{i}]{#2\,#1}
		\newcommand{\unrefinedCell}[2][s\phantom{i}]{#2\,#1}
		\caption{Comparison of runtimes for different block selection heuristics. Heuristics \emph{random} selects a block at random, while \emph{max} and \emph{min} select a block based on how much it increases the size of the winning region. TO indicates that computation did not finish without 60 seconds.}
		\newcommand{\noref}[1]{}
		\label{tab:block_selection}
		{
			\centering		
			\begin{tabular}{lrrrr}
				\toprule
				& \multicolumn{1}{c}{\emph{\emph{{random}}}}
				& \multicolumn{1}{c}{\emph{\emph{{max-$\Delta$}}}}
				& \multicolumn{1}{c}{\emph{\emph{{min-$\Delta$}}}}
				& \multicolumn{1}{c}{\emph{\emph{{min-frontier}}}}
				\\\midrule
    \texttt{generals} &
\statCell[ms]{997} &
\statCell[ms]{976} &
\statCell[ms]{983} &
\statCell[ms]{977} \\
\texttt{station} &
\statCell{7.46} &
\statCell{7.70} &
\statCell{7.64} &
\statCell{7.90} \\
\texttt{philosophers} &
\statCell[ms]{911} &
\statCell[ms]{910} &
\statCell[ms]{895} &
\statCell[ms]{897} \\
\texttt{large\_frontier\_reach} &
TO &
TO &
TO &
TO \\
\texttt{large\_frontier\_safety} &
TO &
TO &
TO &
TO \\
\texttt{almost\_empty\_frontier} &
TO &
TO &
TO &
TO \\
\texttt{centrifuges} &
\statCell{30.07} &
\statCell{23.53} &
\statCell{23.79} &
\statCell{23.32} \\
\texttt{clouds} &
\statCell{1.72} &
\statCell{1.72} &
\statCell{1.73} &
\statCell{1.74} \\
\texttt{complex\_clouds} &
\statCell{3.02} &
\statCell{3.06} &
\statCell{3.05} &
\statCell{3.05} \\
\bottomrule
			\end{tabular}\par
		}
}\end{table}
\paragraph{Choosing refinement blocks.} In every iteration, the set of non-singleton blocks $B$ with a block-switching pair $(C, B)$ is determined. To choose which of these to refine, we provide four heuristics. Heuristics \emph{random} chooses a block randomly, the heuristics \emph{max}-$\Delta$ chooses the block where $|\ggWin(C \cup B) \setminus \ggWin(C)|$ is maximal, heuristics \emph{min}-$\Delta$ chooses the block where it is minimal and heuristics \emph{min-frontier} selects the block with the smallest frontier. The intuition behind the last heuristics is that, for safety and reachability objectives, at least one state in the frontier has positive responsibility by Proposition~\ref{prop:refinement_frontier_useful}, so choosing a small frontier makes it more likely to hit that state.

Table~\ref{tab:block_selection} provides a runtime comparison (for initial partition $n=1$ and refinement heuristics \emph{frontier-random}). It is evident that the difference in performance is small. This is likely because any block that has a block-switching pair needs to be refined eventually.

\paragraph{Refining blocks.} A non-singleton block $B$ with a block-switching pair $(C, B)$ is refined by selecting some states $B' \subseteq B \cap (\ggWin(\ggGame[C \cup B]) \setminus \ggWin(\ggGame[C]))$ and replacing block $B$ by $B'$ and $B \setminus B'$. The heuristics \emph{random} randomly selects a state from $B \cap \ggWin(\ggGame[C \cup B]) \setminus \ggWin(\ggGame[C])$. The remaining heuristics are variants of the \emph{frontier} heuristics introduced in Definition~\ref{def:frontier}. \emph{Frontier-random} randomly selects a block from the frontier, \emph{frontier-max} selects the block with the most transitions to $\ggWin(\ggGame[C]) \cup \tsStates \setminus \ggWin(\ggGame[C \cup B])$, \emph{frontier-losing} selects the block with the most transitions to $ \tsStates \setminus \ggWin(\ggGame[C \cup B])$ and \emph{frontier-winning} selects the block with the most transitions to $\ggWin(\ggGame[C])$.

Table~\ref{tab:splitting_heuristics} shows the results (for initial partition $n=1$ and block selection heuristics \emph{random}). Heuristics \emph{random} times out on every model, indicating that \emph{frontier} is much more suitable for identifying responsible states. The different variants of \emph{frontier} perform almost identical on most models. The only exception are \texttt{large\_frontier\_reach} and \texttt{large\_frontier\_safety}, where only \emph{frontier-win} and \emph{frontier-los}, respectively, are able to finish computation before the timeout. The respective heuristics directly identify which states in the frontier are responsible, whereas the other heuristics randomly choose states, half of which have no responsibility.

\begin{table}[b]{
	\setlength{\tabcolsep}{11.1pt}
	\newcommand{\statCell}[2][s\phantom{i}]{\hspace{-0.8em}#2\,#1}
	\newcommand{\unrefinedCell}[2][s\phantom{i}]{#2\,#1}
	\caption{Comparison of runtimes for different splitting heuristics. TO indicates that computation did not finish without 60 seconds.}
	\label{tab:splitting_heuristics}
	{\centering
		\begin{tabular}{lrrrrr}
			\toprule
			& & \multicolumn{4}{c}{\emph{frontier}} \\
			\cmidrule(lr){3-6}
			& \multicolumn{1}{c}{\hspace{-2em}\emph{{random}}}
			& \multicolumn{1}{c}{\emph{{-random}}}
			& \multicolumn{1}{c}{\emph{{-max}}}
			& \multicolumn{1}{c}{\emph{{-losing}}}
			& \multicolumn{1}{c}{\emph{{-winning}}}
			\\\midrule
			\texttt{generals} &
			TO &
			\statCell[ms]{991} &
			\statCell[ms]{988} &
			\statCell{1.01} &
			\statCell[ms]{981} \\
			\texttt{station} &
			TO &
			\statCell{7.88} &
			\statCell{7.72} &
			\statCell{7.62} &
			\statCell{7.64} \\
			\texttt{philosophers} &
			TO &
			\statCell[ms]{901} &
			\statCell[ms]{907} &
			\statCell[ms]{900} &
			\statCell[ms]{911} \\
			\texttt{large\_frontier\_reach} &
			TO &
			TO &
			TO &
			\statCell[ms]{295} &
			TO \\
			\texttt{large\_frontier\_safety} &
			TO &
			TO &
			TO &
			TO &
			\statCell[ms]{209} \\
			\texttt{almost\_empty\_frontier} &
			TO &
			TO &
			TO &
			TO &
			TO \\
			\texttt{centrifuges} &
			TO &
			\statCell{29.68} &
			\statCell{23.18} &
			\statCell{23.24} &
			\statCell{23.37} \\
			\texttt{clouds} &
			TO &
			\statCell{1.72} &
			\statCell{1.72} &
			\statCell{1.72} &
			\statCell{1.72} \\
			\texttt{complex\_clouds} &
			TO &
			\statCell{3.05} &
			\statCell{2.98} &
			\statCell{3.04} &
			\statCell{2.95} \\
			\bottomrule
		\end{tabular}
		\par}
}\end{table}

\subsection{Case study: Sample analysis}
\label{sec:casestudyone}

This section analyses a model of a medical lab that analyses samples to demonstrate the efficacy of combining module-based grouping and refinement.

The system starts with a given number of clean and infected samples. There are multiple centrifuges that perform an identical analysis of the samples, which are modelled as individual \textsc{Prism} modules. Each centrifuge iteratively performs three steps: It receives a sample (either clean or infected) from a central supply system. The sample is applied to a sample strip and the speed at which the sample soaks into the strip is selected based on the type of sample and a probabilistic choice: An infected sample advances 2 or 3 centimetres per minute, whereas a clean sample advances by 0 or 1 centimetre. After 3 minutes have elapsed, the centrifuge determines the type of sample based on whether it soaked at least 4 centimetres of the test strip. It reports the result to a centralised counter.

To ensure all centrifuges operate correctly, we check the property that eventually, all samples are processed and the number of clean and infect results matches the number of clean and infected samples initially present.

Running a model checker on the model reveals that the property is not fulfilled. The number of clean results might be higher than the number of clean samples initially present. As the model specification has hundreds of lines and the model hundreds of thousands of states, manual analysis is challenging. Previously, applying responsibility techniques was also intractable, as individual-state-analysis is only tractable for up to roughly 30 states. Similarly, module-based analysis was intractable: For 10 centrifuges, there are 53 modules and synchronising actions. Group-based analysis using the brute-force algorithm is only feasible for up to roughly 30 groups.

\begin{figure}
	\centering
	\includegraphics[width=0.8\textwidth]{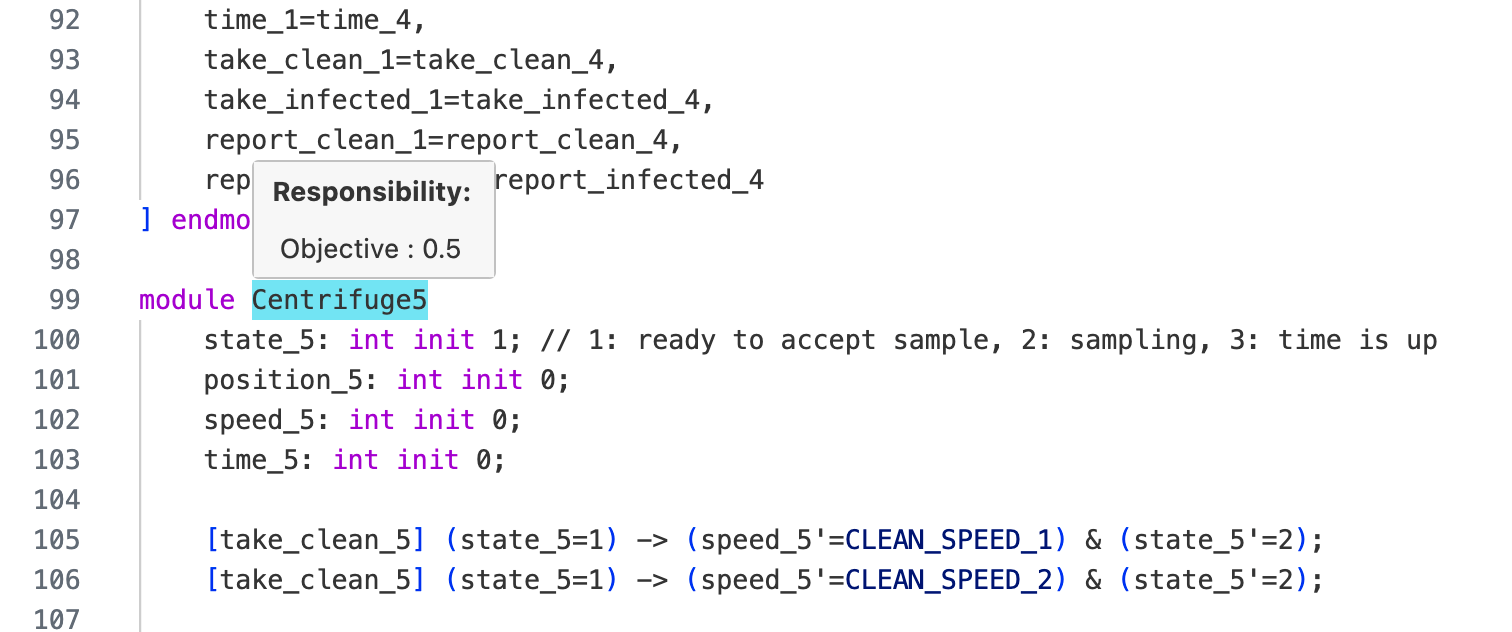}
	\caption{A screenshot of PMC-Vis displaying module responsibility in its code editor.}
	\label{fig:pmc_vis_code}
\end{figure}

However, combining module-based grouping with the refinement algorithm is fruitful: Computing responsibility takes less than 30 seconds and reveals that the scheduler and Centrifuge 5 both have a responsibility value of $\frac{1}{2}$. \Cref{fig:pmc_vis_code} shows how PMC-VIS highlights Centrifuge 5 in its code view. This suggests that Centrifuge 5 is faulty (and the scheduler has some responsibility, as it could ensure Centrifuge 5 never processes a sample). Analysis of the code of Centrifuge 5 reveals the error: It uses \texttt{<=} instead of \texttt{=} when deciding whether 3 minutes have passed, leading to a potential early abort.

\section{Conclusion}

We have extended the definition of backward responsibility \cite{BackwardRespTS} to reachability, Büchi and parity objectives. The results for safety objectives transfer to reachability objectives, while the optimistic case is more difficult for Büchi objectives. Nonetheless, by exploiting the structure of minimal winning coalitions, the set of responsible states can be computed in polynomial time. For parity objectives, no polynomial algorithm is known. For all objective classes, we have additionally provided upper bounds for the positivity, threshold and computation problem.

Secondly, we have presented a refinement algorithm that computes the set of states with positive responsibility. Our experimental evaluation shows that refinement enables analysis of models where previous techniques timed out. We have provided a theoretical foundation for the \emph{frontier} heuristics and shown that it outperforms random choice on several models. In a case study, we have demonstrated how responsibility values can be used to pin-point bugs in complex systems.
\bibliography{bib}
\end{document}